\documentclass[journal,draftclsnofoot,onecolumn,12pt,twoside]{IEEEtranTCOM} 
\renewcommand{\baselinestretch}{1.45}

\usepackage{cite}
\usepackage{latexsym}
\usepackage{verbatim, amsbsy}
\usepackage{algorithm}
\usepackage[noend]{algpseudocode}

\usepackage[dvips]{graphicx}
\usepackage{subfigure}
\usepackage{color}
\usepackage{amsmath}                      
\usepackage{amssymb}
\usepackage{dsfont}
\usepackage{theorem}
\usepackage{amsfonts}
\usepackage{times}
\usepackage[english]{babel}
\usepackage[dvips]{graphicx}
\usepackage[normalem]{ulem}
\usepackage{bm}
\usepackage{array}
\usepackage{tikz}
\usetikzlibrary{arrows,positioning,shapes.geometric}

\usepgflibrary{shapes.geometric}

\newcommand{\ben}{\begin{enumerate}}
\newcommand{\een}{\end{enumerate}}
\newcommand{\be}{\begin{equation}}
\newcommand{\ee}{\end{equation}}
\newcommand{\bea}{\begin{eqnarray}}
\newcommand{\eea}{\end{eqnarray}}
\newcommand{\bc}{\begin{cases}}
\newcommand{\ec}{\end{cases}}
\newcommand{\bi}{\begin{itemize}}
\newcommand{\ei}{\end{itemize}}

\newcommand{\spa}{\,\,\,\!\!}

\newtheorem{teo}{Theorem}[section]

\newtheorem{lem}{Lemma}[section]
\newcommand{\eq}[1]{(\ref{#1})}
 
 \newtheorem{rem}{Remark}[section]

\def\de{\mathrm{d}}

\newcommand{\figw}{0.46\columnwidth}

\begin{document}
\title{Optimal Context Aware Transmission Strategy for   non-Orthogonal  D2D Communications}

\author{\authorblockN{{\bf Federico Librino}$^\star$, {\bf Giorgio Quer}$^\diamond$}\\
$^\star$Italian National Research Council -- 56124 Pisa, Italy,\\
$^\diamond$ University of California San Diego -- La Jolla, CA 92093, USA.
\thanks{This work has been submitted to the IEEE for possible publication. Copyright may be transferred without notice, after which this version may no longer be accessible.}
}

\maketitle

\thispagestyle{empty}

\begin{abstract}
The increasing traffic demand in cellular networks has recently led to the investigation of new strategies to save precious resources like spectrum and energy. A possible solution employs direct device-to-device (D2D) communications, which is particularly promising  when the two terminals involved in the communications are located in close proximity. 
The D2D communications should coexist with other transmissions, so they must be careful scheduled in order to avoid harmful interference impacts.
In this paper, we analyze how a distributed context-awareness, obtained by observing few local channel and topology parameters, can be used to adaptively exploit D2D communications. We develop a rigorous theoretical analysis to quantify the balance between the gain offered by a D2D transmission, and its impact on the other network communications.
Based on this analysis, we derive two theorems that define the optimal strategy to be employed, in terms of throughput maximization, when a single or multiple transmit power levels are available for the D2D communications.
We compare this strategy to the state-of-the-art in the same network scenario, showing how context awareness can be exploited to achieve a higher sum throughput and an improved fairness.
\end{abstract}

\section{Introduction}
\label{sec:introduction}
Cellular networks have undergone a widespread diffusion and a rapid development in the past decade. The fourth generation (4G) has introduced several improvements, ranging from better coding/decoding devices to novel communication strategies, as well as a smarter way of exploiting time and frequency resources.
These improvements have led to more reliable, faster and more secure communications, which are fundamental to support the large spread of mobile applications, now fostered by the pervasive utilization of social media. Correspondingly, a huge increase in the amount of traffic load has been observed and is expected to continue at even higher pace in the next future.
Sharing photos, music, videos and other types of multimedia data is becoming part of the everyday life for millions (and soon billions) of people.

In this scenario, the big challenge for the next generation of cellular networks (5G) lies in finding new ways to handle this huge amount of communications with the required Quality of Service (QoS). On one side, new frequency bands might soon become available, but spectrum remains a finite resource. On the other side, a promising approach for 5G networks is to devise new ways to reuse the available spectrum as much as possible, by allowing Non Orthogonal Multiple Access (NOMA).

In particular, the possibility of establishing direct device-to-device (D2D) communications is now attracting significant attention  \cite{Rel3GPP,over3GPP,Rel3GPPbis,com5Gnet,comD2Dcell,BookWire}. This technique might be particularly advantageous when traffic is to be exchanged locally, that is, between two terminals in close proximity. Examples of this kind of scenario are certain gaming applications, file exchange between colocated user equipments (UEs), or data downloads from low-power wireless sensors in a smart city~\cite{smartcity}.
In this situations, the standard cellular communication mode (two hops through the Base Station, BS)
might lead to a waste of resources. Instead, the two UEs could leverage the short direct link using low power communications, possibly reducing latency, costs, and energy.

Several strategies have been proposed to allow D2D communications in the presence of local traffic~\cite{comD2Dcell}. If unused spectrum resources are available, D2D communications can be performed over a dedicated, orthogonal channel, without interfering with other concurrent cellular transmissions.
However, if the spectrum is fully allocated, 
the spectrum sharing principle can be adopted, and a proper communications scheduling is necessary in order to avoid overwhelming interference. In most cases, this scheduling is organized by the BS, based on the channel state information (CSI) about all the involved channels~\cite{mumimo,comundcell}.
While this approach can achieve an optimal spectrum allocation, its drawback lies in the considerable overhead required to gather all the necessary information at the BS.

Distributed strategies have been considered as well. In this case, the UEs can exploit D2D communications autonomously, based on some geographic considerations~\cite{modpowcon}. These solutions alleviate the network from the burden of the overhead required by a centralized solution, but considering only geographic information leads to over-restrictive constraints.

In this work, instead, we opt for a different, distributed and adaptive approach based on situational awareness. We focus on the non orthogonal scenario, where the uplink resources are utilized for D2D communications.
The main idea is to make the D2D source collect some local information about the surrounding communications, infer the quality of the D2D transmission, as well as the amount of generated interference, and use this information to decide whether to start a D2D transmission or not. The result is a dynamic strategy, which exploits D2D communications only where and when this is feasible.

In order to keep the interference level under control, the BS still retains the possibility, if the QoS of the uplink communications falls below a predefined threshold, to block the D2D sources for a certain amount of time. This control mechanism forces a D2D source to carefully choose the best options to transmit, in order to optimize its performance.

In the paper, we analyze in depth the achievable throughput of our approach in a single-cell scenario, by deriving two theorems which describe the optimal strategy for the D2D source. More specifically, the organization and the contributions of this paper are summarized in the following.

\noindent \textbullet ~~We overview in Sec.~\ref{sec:relatedwork} the most recent 3GPP releases in D2D communications and the related works in which spectrum resources are shared between cellular and D2D UEs.

\noindent \textbullet ~~In Sec.~\ref{sec:tx_policy_1} we describe our system model by means of a Markov Decision Process; moreover, we prove a theorem which gives the general form of the optimal strategy to be employed by the D2D source in our scenario when a single transmit power level is available. The application of this theorem in the context of cellular networks is provided in Sec.~\ref{sec:D2Dpract_sinpow}, where a practical strategy is also introduced.

\noindent \textbullet ~~In Sec.~\ref{sec:opt_mult_act} we extend the analysis to the case in which multiple power levels can be dynamically selected  by the D2D source. A second theorem is stated and proved to illustrate the general form of the optimal strategy in this scenario, with a practical implementation described in Sec.~\ref{sec:mulpowlev}.

\noindent \textbullet ~~Finally, in Sec.~\ref{sec:results} we compare the performance of the strategy based on our approach with that of a distributed strategy recently appeared in the literature, showing how context awareness is beneficial in terms of throughput and fairness for both the D2D sources and the other cellular users.

\noindent Sec.~\ref{sec:conclu} concludes the paper and presents some future works.

\section{Related Work}
\label{sec:relatedwork}
Direct communications among mobile terminals have been envisioned by 3GPP as a promising way to improve network performance. D2D proximity services can in fact either reduce the amount of traffic handled by the BSs, or provide service beyond cellular coverage and/or in emergency scenarios, where the core network may be unavailable~\cite{Rel3GPP, Rel3GPPbis}.
Establishig and maintaining D2D connections entails a set of technical challenges, including peer discovery, resource allocation, interference management and synchronization, which are presented and discussed in~\cite{over3GPP}. Multiple-input-multiple-output (MIMO) D2D communications are investigated in~\cite{comD2Dcell}, whereas some pricing models are illustrated in~\cite{com5Gnet}.
An opportunistic multi-hop forwarding technique is presented in~\cite{dyngra}: the main aim here is to extend cellular coverage, by letting a mobile terminal forward the data packets to/from another terminal that is not within the range of a BS. Similar works along this line of research seek to reduce the density of the BSs.

Our work, on the contrary, does not aim at extending coverage, but instead at allowing direct communications between terminals in close proximity.
Authors in~\cite{intraclu} propose a technique to organize nodes into clusters, by means of a centralized scheduling, and exploit D2D communications over an orthogonal channel.
Conversely, we focus on a scenario where non-orthogonal spectrum sharing between terminals transmitting to the BS and D2D sources is employed. In this kind of scenario, two main approaches have been proposed in the state-of-the-art. The former is to let D2D sources transmit only on temporarily free channels (overlay), thus causing no extra interference; the latter is to allow D2D transmissions also on already utilized channels, but limiting the interference impact on the other ongoing communications (underlay).
An overlay scheme is developed in~\cite{cogene}, where D2D sources exploit the energy harvested by surrounding radio communications. Stochastic geometry tools are instead utilized in~\cite{concog} to analyze the performance of both overlay and underlay schemes in terms of network connectivity and coverage probability.

An interesting underlay approach comparable to ours is proposed in~\cite{modpowcon}. Here, D2D communications are performed using uplink resources, and employing power control, in order to limit interference. Furthermore, D2D connections are allowed only between terminals located in close proximity, as in our work. However, the source terminal decides whether to transmit directly to its destination or to rely on the BS based only on topological information.
This scheme avoids the need for channel sensing overhead, but it lacks the fundamental adaptivity of our approach. We briefly describe it in Sec.~\ref{sec:results}, where we compare its performance with that of our proposed strategy.

Overall, the main difference between our approach and the existing ones lies in our strategy to mitigate the interference. We do not rely neither on a geographic based criterion, as in~\cite{modpowcon}, which is easy to implement in a distributed fashion but is static and often over-restrictive, nor on a centralized optimization problem, as in~\cite{intraclu, comundcell,resouopt,groupart}, which achieves optimal solutions, but needs full channel state information over all the involved channels.

Conversely, we create a distributed situational awareness by proper observations of some channel parameters, and exploit it through our analytical results to choose when and how a D2D connection can be established. We have already investigated the concept of situational awareness in a multi-cell scenario in~\cite{ourglobe}. In that work, however, context awareness is based on statistical information, and decisions are taken based on the output of properly designed Bayesian Networks, without seeking to find the optimal solution.

Several other strategies have appeared in the literature to permit D2D communications through an underlay approach. Authors in~\cite{coopintcan} defined a scheme to forward interference by means of D2D communications, in order to make it easier to apply interference cancellation schemes; similarly, in~\cite{relimpro} the BS relays the D2D communications to allow interference cancellation at the receiving nodes.
In \cite{groupart}, graph theory is used to divide mobiles into subgroups, and a throughput maximization is attained by employing multiuser detection and solving an iterative optimization algorithm. In~\cite{contract} contract theory is leveraged to study the incentives to be granted to potential D2D users.
Finally, in~\cite{resousha}, various sharing schemes, both orthogonal and non-orthogonal, are investigated in a Manhattan grid topology based on the solution of a sum-rate optimization problem.

\section{Optimal Transmission Policy in a Multi-Tier Network}
\label{sec:tx_policy_1}

We first describe our framework in the most general scenario, in which there is one licensed owner of the spectrum resource, $U$, transmitting to $B$, and one user $S$ that is attempting a transmission towards $D$ over the same spectrum. Both users are backlogged, the time is slotted ($t\in\mathbb{N}$), and a transmission can be performed in a single time slot.

The transmissions from $S$ at time $t$ has a probability of success $p_t$, which depends on the channel conditions and the interference level at the destination. The licensed owner of the spectrum, $U$, allows the transmissions from $S$ only if it does not create a harmful interference. Indeed, $U$ can measure the level of interference from $S$ and, if this level is above a certain threshold, $U$ can block $S$ by denying its access to the resources for the subsequent $W$ time slots. The probability of blocking the transmissions of $S$ at time $t$ is denoted by $q_t$.

Assuming $p_t$ and $q_t$ are known, the secondary user $S$ can decide at each time slot $t$ if to transmit a new data packet, or to defer transmission to avoid the risk of a blockage. 
In the following, we derive an optimal transmission policy for $S$, with the strategy to be used in order to maximize its overall throughput. In Sec.~\ref{sec:D2Dpract_sinpow}, we will specify how $S$ can calculate  the values of $p_t$ and $q_t$ from the available local information.

\subsection{Modeling the Transmissions as a Stochastic Process}
Both $p_t$ and $q_t$ change as a function of time, depending on the channel conditions, so they can be modeled as two stochastic processes $\{p\} = \{P_t, t\in\mathbb{N}\}$ and $\{q\} = \{Q_t, t\in\mathbb{N}\}$, which may also be correlated. At a specific time step, the values of $p_t$ and $q_t$ are given by the realizations of the random variables $P_t$ and $Q_t$, respectively.
Let us start by assuming that these two processes are Markovian and stationary, so we can express the joint probability as $\phi_{P,Q}(p_t, q_t |p_{t-1},q_{t-1})$. 
We can analyze the strategy of the user $S$ with a Markov decision process (MDP).

\subsection{The MDP model}
The set of actions allowed for $S$ is $\mathcal{A} = \{T,H\}$, i.e., it can transmit ($T$) a packet or halt ($H$) the transmissions to avoid a possible blockage, depending on the state of the system. 

The state depends on the probability of success $p_t$, the probability of being silenced $q_t$, and the current blockage $\Lambda_t\in\{0, 1, 2, \ldots, W\}$, where $\Lambda_t = 0$ means that $S$ is free to transmit, while $\Lambda_t = i>0$ means that $S$ has been blocked, and must remain silent for this time slot, and for the following $W-i$ ones. We assume that $p_t$ and $q_t$ can assume only discrete values, with $p_t\in\mathcal{P}$ and $q_t\in\mathcal{Q}$, and $\mathcal{P},\mathcal{Q}\subset[0,1]$ (this assumption will be relaxed in the following). Thus, the state can be defined as the triplet $\bm s_t = (\Lambda_t, p_t, q_t)$. 

The transition probabilities $\mathbb{P}\left[ \bm s_{t+1}|\bm s_t\right]$ can be calculated case by case. $S$ can halt the transmission by choice (if $\Lambda_t=0$), or it can be forced to do it (if $\Lambda_t>0$). In both cases
\begin{equation}
 \mathbb{P}_{a=H}\left[ (\Lambda_{t+1}, p_{t+1}, q_{t+1})|(\Lambda_t, p_t, q_t)\right] =
\phi_{P,Q}(p_{t+1}, q_{t+1}|p_t,q_t)  \; ,
\label{tran_Ho}
\end{equation}
where $\Lambda_{t+1}=0$ if $\Lambda_t=0$ (if not in a blockage period), $\Lambda_{t+1}=\Lambda_t +1$ if $1 \leq \Lambda_t < W$ (if $S$ was blocked), or  $\Lambda_{t+1}=0$ if $\Lambda_t=W$ (if it was the last time slot of the blockage).

$S$ can decide to transmit only if $\Lambda_t=0$. In this case, a blockage can be triggered
\begin{equation}
 \mathbb{P}_{a=T}\left[(\Lambda_{t+1}, p_{t+1}, q_{t+1})|(\Lambda_t=0, p_t, q_t)\right] =
 \begin{cases}
  (1-q_t)\phi_{P,Q}(p_{t+1}, q_{t+1}|p_t,q_t) & \text{if } \Lambda_{t+1} = 0 \\
  q_t\phi_{P,Q}(p_{t+1}, q_{t+1}|p_t,q_t) & \text{if } \Lambda_{t+1} = 1 \; .\\
 \end{cases}
 \label{tran_Tx}
\end{equation}

In this MDP model, a reward $\mathbb{R}_a(s_t)$ depends only on the current state and the action taken. We define the reward to be equal to $1$ if a packet is delivered, $0$ otherwise. If $a = T$ we have 
$\mathbb{R}_{a=T}\left[(\Lambda_{t}=0, p_t, q_t)\right] = p_t $, 
while $\mathbb{R}_{a=H}\left[(\Lambda_{t}, p_t, q_t))\right] = 0$ otherwise. 

In the MDP model, the optimal policy $\bm{\sigma}^*$ selects one action for each state ($\sigma(\bm s)\in\mathcal{A}$) in order to maximize the expected total reward over an infinite horizon. In order to calculate $\bm{\sigma}^*$, we should initialize the vector of values $\bm{V}$, which contains the expected reward starting from each possible state. A possible choice for $\bm{V}^0$ is the null vector. 
Then, we should define also an initial policy $\bm{\sigma}^0$. The length of both vectors is equal to the number of possible states $N_s = |\mathcal{P}||\mathcal{Q}|(W+1)$.

We can find the optimal policy $\bm{\sigma}^*$ with dynamic programming~\cite{bertsy} by updating the values of $\bm{\sigma}$ and $\bm{V}$ for every state $\bm s$, starting from $\bm{\sigma}^0$ and $\bm{V}^0$, as follows:
\be
 \sigma^{k+1}[\bm s]  =  \arg\max_a\left\{\sum_{\bm s'}\mathbb{P}_a\left[\bm s'|\bm s\right]\left(\mathbb{R}_a\left[\bm s\right] + \gamma V^k[\bm s']\right)\right\} 
 \label{updapi} 
 \ee
 \be
 V^{k+1}[\bm s]  =  \sum_{\bm s'}\mathbb{P}_{\sigma^{k+1}[\bm s]}[\bm s'|\bm s]\left(\mathbb{R}_{\sigma^{k+1}[\bm s]}\left[\bm s\right] + \gamma V^k[\bm s']\right) \; ,
 \label{updaV}
\ee
where  $\gamma\in(0,1)$ is the discount factor, and $k$ represents the iteration step. The iteration finishes with convergence, i.e., when all elements of $\bm{V}^{k+1} - \bm{V}^{k}$ have an absolute value below $\epsilon \ll 1$.

\subsection{Time uncorrelation}

In a real system, we often have negligible time correlation for the channel fading coefficients, e.g., if $S$ is transmitting only in a subset of non-consecutive time slots. In this case, $\phi_{P,Q}(p_{t+1}, q_{t+1})$ does not depend on the current values of $p$ and $q$, and the MDP model can be further simplified.

This simplification allows us to partition the set of states $\mathcal{S}$ into two disjoint subsets: $\mathcal{S}_F = \{\bm s_t \in\mathcal{S}:\Lambda_t = 0\}$, so we can write $\bm s_t =(p_t,q_t)$, and $\mathcal{S}_B = \{\bm s\in\mathcal{S}:\Lambda_t>0\}$, where $\bm s_t = (\Lambda_t)$. The last simplification is possible since no transmission is allowed if $\Lambda_t>0$, and the future values of $p_{t+1}$ and $q_{t+1}$ do not depend on their current values. In particular, if $1 \leq \Lambda_t < W$, then the next state is $\bm s_{t+1} = (\Lambda_t+1)$ with probability 1. The number of states becomes $N_s= |\mathcal{P}||\mathcal{Q}| + W$.

Fig.~\ref{diag:exaMDP} depicts an example of the MDP model with 4 discrete values for $p$ and $q$, i.e.,  $\mathcal{P} = \{p^{(1)},p^{(2)},p^{(3)},p^{(4)}\}$ and $\mathcal{Q} = \{q^{(1)},q^{(2)},q^{(3)},q^{(4)}\}$, and $W=3$. The initial state is $\bm s = (\Lambda = 0, p = p^{(2)}, q = q^{(4)})$, a solid line indicates a transition without transmissions, while a dashed line indicates a transition with a transmission attempt. When the system is in state $\bm s = (\Lambda = 0, p = p^{(4)}, q = q^{(3)})$, a transmission is attempted and a blockage is triggered, bringing the system to the blockage state $\bm s = (\Lambda = 1)$.

With these assumptions, it is possible to derive analytically the optimal decision policy $\bm \sigma^*$ as follows.
When the algorithm has converged ($\bm{V}^{k+1} = \bm{V}^k = \bm{V}$), the optimal action $\sigma[\bm s]$ is known for every state. This action can be to transmit ($a=T$) or not ($a=H$) for any state $\bm s\in\mathcal{S}_F$, whereas it must be $a= H$ for the states in $\mathcal{S}_B$.

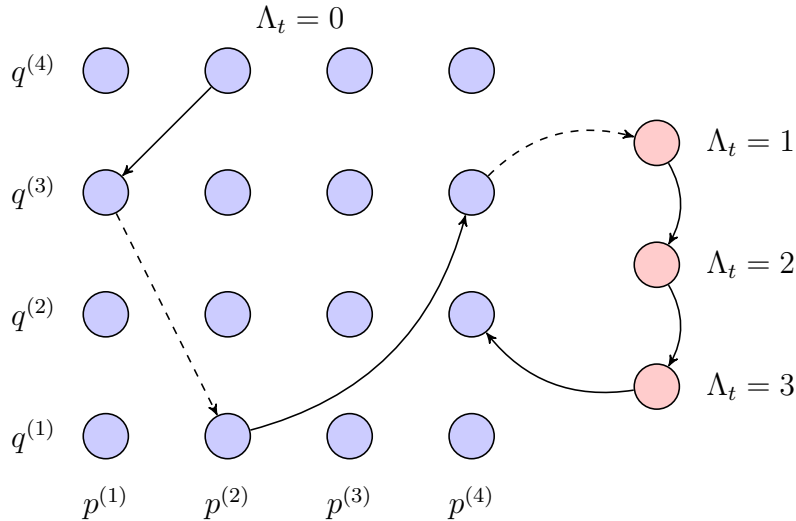
\begin{figure}
 \begin{center}
  \begin{tikzpicture}[>=stealth',semithick,auto]	 
  \tikzset{mypoint/.style = {draw, shape=circle, text centered, align=center, minimum size=0.6cm, fill=blue!20},
	   myblo/.style = {draw, shape=circle, text centered, align=center, minimum size=0.6cm, fill=red!20},
	   mytext/.style = {text centered, align=center, minimum size=0.6cm}
		        }
		        
    \node [mypoint] (p1q4) {};
    \node [mypoint, right=1cm of p1q4] (p2q4) {};
    \node [mypoint, right=1cm of p2q4] (p3q4) {};
    \node [mypoint, right=1cm of p3q4] (p4q4) {};
    
    \node [mypoint, below=1cm of p1q4] (p1q3) {};
    \node [mypoint, right=1cm of p1q3] (p2q3) {};
    \node [mypoint, right=1cm of p2q3] (p3q3) {};
    \node [mypoint, right=1cm of p3q3] (p4q3) {};
    
    \node [mypoint, below=1cm of p1q3] (p1q2) {};
    \node [mypoint, right=1cm of p1q2] (p2q2) {};
    \node [mypoint, right=1cm of p2q2] (p3q2) {};
    \node [mypoint, right=1cm of p3q2] (p4q2) {};
    
    \node [mypoint, below=1cm of p1q2] (p1q1) {};
    \node [mypoint, right=1cm of p1q1] (p2q1) {};
    \node [mypoint, right=1cm of p2q1] (p3q1) {};
    \node [mypoint, right=1cm of p3q1] (p4q1) {};
    
    \node [ , below=0.5cm of p4q4] (tmp) {};
    \node [myblo, right=2cm of tmp] (b1) {};
    \node [myblo, below=1cm of b1] (b2) {};
    \node [myblo, below=1cm of b2] (b3) {};
    
     \draw[->] (p2q4) to node[auto](step1){} (p1q3);
     \draw[->, dashed] (p1q3) to node(step2){} (p2q1);
     \draw[->] (p2q1) to [bend right = 30] node[auto](step3){} (p4q3);
     \draw[->, dashed] (p4q3) to [bend left = 30] node(step4){} (b1);
     \draw[->] (b1) to [bend left = 30] node(step5){} (b2);
     \draw[->] (b2) to [bend left = 30] node(step6){} (b3);
     \draw[->] (b3) to [bend left = 30] node(step5){} (p4q2);
     
    \node [mytext, right=0.2cm of b1] (nb1) {$\Lambda_t = 1$};
    \node [mytext, right=0.2cm of b2] (nb2) {$\Lambda_t = 2$};
    \node [mytext, right=0.2cm of b3] (nb3) {$\Lambda_t = 3$};
    
    \node [mytext, below=0.2cm of p1q1] (np1) {$p^{(1)}$};
    \node [mytext, below=0.2cm of p2q1] (np2) {$p^{(2)}$};
    \node [mytext, below=0.2cm of p3q1] (np3) {$p^{(3)}$};
    \node [mytext, below=0.2cm of p4q1] (np4) {$p^{(4)}$};
    
    \node [mytext, left=0.2cm of p1q1] (nq1) {$q^{(1)}$};
    \node [mytext, left=0.2cm of p1q2] (nq2) {$q^{(2)}$};
    \node [mytext, left=0.2cm of p1q3] (nq3) {$q^{(3)}$};
    \node [mytext, left=0.2cm of p1q4] (nq4) {$q^{(4)}$};
    
    \node [ , right=0.5cm of p2q4] (tmp2) {};
    \node [mytext, above=0.2cm of tmp2] (lam0) {$\Lambda_t = 0$};
    
    \end{tikzpicture}
    \caption{Evolution of the MDP when $\{p\}$ and $\{q\}$ are stationary and non time correlated. The states on the left are those with $\Lambda_t = 0$, each one identified by the value of $P_t$ and $Q_t$. Dashed arrows correspond to action $T$, while solid arrows to action $H$. Here it is assumed that $|\mathcal{P}| = |\mathcal{Q}| = 4$ and $W=3$.}
    \label{diag:exaMDP}
    \vspace{-1cm}
  \end{center}
\end{figure}

Let us consider first the iteration from those states $\bm s \in \mathcal{S}_F$ for which $\sigma[\bm s] = H$. We have that $\Lambda_{t+1} = 0$, and the reward is $0$. For these states we can write 
\be  
V[\bm s] = \gamma\sum_{\bm s'\in{\mathcal{S}_F}}\mathbb{P}_H[\bm s'|\bm s]V[\bm s'] \\
=  \gamma\sum_{p'\in\mathcal{P},q'\in\mathcal{Q}}\phi_{P,Q}(p', q')V[(p', q')]
\label{valfin_HS}
\ee
where we used \eq{tran_Ho} in the time uncorrelated case. For those states $\bm s \in \mathcal{S}_F$ for which $\sigma[\bm s] = T$ we also have an expected reward, as well as a probability of being blocked
\begin{eqnarray}
 \bm{V}[\bm s] & = & \sum_{\bm s'\in{\mathcal{S}}}\mathbb{P}_T[\bm s'|\bm s]\left(p + \gamma V[\bm s']\right) \nonumber\\
 & = & p + \gamma\sum_{\bm s'\in{\mathcal{S}}}\mathbb{P}_T[\bm s'|\bm s]V[\bm s'] \nonumber\\
 & = & p + \gamma\sum_{\bm s'\in{\mathcal{S}_F}}\mathbb{P}_T[\bm s'|\bm s]V[\bm s'] + \gamma\sum_{\bm s'\in{\mathcal{S}_B}}\mathbb{P}_T[\bm s'|\bm s]V[\bm s'] \nonumber\\
 & = & p + \gamma(1-q)\sum_{p'\in\mathcal{P},q'\in\mathcal{Q}}\phi_{P,Q}(p', q')V[(p', q')] + \gamma qV[(\Lambda = 1)] \; .
 \label{valfin_TS}
\end{eqnarray}

The update equations for those states $\bm s \in \mathcal{S}_B$ can be calculated by observing that from state  $\bm s = (\Lambda=i)$, with $1 \leq i \leq W $, the only possible action is $a=H$, thus no reward is involved. These equations can hence be written as
\begin{equation}
 V[(\Lambda=i)] =
 \begin{cases}
  \gamma V[(\Lambda=i+1)] & \mbox{if } i < W \\
  \gamma\sum_{p'\in\mathcal{P},q'\in\mathcal{Q}}\phi_{P,Q}(p', q')V[(p', q')] & \mbox{if } i = W \; ,
 \end{cases}
\end{equation}
since the case of $i=W$ corresponds to the end of the blockage. The update from all the states in $\mathcal{S}_B$ is deterministic, so we can rewrite this equation as 
\begin{equation}
 V[(\Lambda=1)] = \gamma^W\sum_{p'\in\mathcal{P},q'\in\mathcal{Q}}\phi_{P,Q}(p', q')V[( p', q')] \; ,
\end{equation}
and combine this results with \eq{valfin_HS} and \eq{valfin_TS}, obtaining
\begin{equation}
 V[(0,p,q)] =
 \begin{cases}
  \displaystyle
  p + \gamma(1-q+\gamma^Wq)\sum_{p'\in\mathcal{P},q'\in\mathcal{Q}}\phi_{P,Q}(p', q')V[( p', q')] & \mbox{if } \sigma[(p,q)] = T\\
  \displaystyle
  \gamma\sum_{p'\in\mathcal{P},q'\in\mathcal{Q}}\phi_{P,Q}(p', q')V[(p', q')] & \mbox{if } \sigma[(p,q)] = H \; .
 \end{cases}
 \label{sol_disc}
\end{equation}

Once the convergence is reached, the final value of $\bm{\sigma}$ is the optimal strategy, which grants the highest reward. \eq{sol_disc} can be rewritten as
\be
 V[(p,q)] = \max\left(p + \gamma(1-q+\gamma^Wq)C_d, \gamma C_d\right) \; ,
 \label{sol_disc_max}
\ee
where we define $C_d$ as the value of the summation over $p'$ and $q'$. We notice that if
we increase the cardinality of both $\mathcal{P}$ and $\mathcal{Q}$ to the limit for $|\mathcal{P}|, |\mathcal{Q}|\rightarrow\infty$, we obtain the continuous case. The PMF $\phi_{P,Q}(p,q)$ becomes a probability density function (PDF) defined over $[0,1]^2$.
$V[(p, q)]$ becomes a function of two continuous variables, $v(p,q)$, which is defined for $0\leq p,q\leq 1$. The vector of optimal actions $\bm \sigma^*$ can be substituted by a binary function $\mu(p,q)$, with $\mu(p,q)=1$ if the optimal action is $T$, and $\mu(p,q)=0$ otherwise. Replacing summations with integrals, we get
\begin{equation}
 v(p,q) =
 \begin{cases}
  \displaystyle
  p + \gamma(1-q+\gamma^Wq)\int_0^1\int_0^1\phi_{P,Q}(p', q')v(p', q')\de p'\de q' & \text{if } \mu(p,q) = 1\\
  \displaystyle
  \gamma\int_0^1\int_0^1\phi_{P,Q}(p', q')v(p', q')\de p'\de q' & \text{if } \mu(p,q) = 0 \; .
 \end{cases}
 \label{sol_cont}
\end{equation}
Similarly to the discrete case, the double integral is a constant $C>0$, since both $\phi_{P,Q}(p,q)$ and $v(p,q)$ are non negative functions:
\begin{equation}
 C = \int_0^1\int_0^1\phi_{P,Q}(p', q')v(p', q')\de p'\de q' \; .
 \label{defC}
\end{equation}

In order to solve \eq{sol_cont} and find the optimal action policy $\mu(p,q)$, we introduce the following lemma.
\begin{lem}
 The optimal action function $\mu(p,q)$ is non decreasing in $p$ and non increasing in $q$
 \begin{eqnarray}
  \mu(p, \bar{q}) & \leq & \mu(p+h, \bar{q}), \quad \forall  h>0, \forall \bar{q}\in[0,1]\label{condp}\\
  \mu(\bar{p}, q) & \geq & \mu(\bar{p}, q+h), \quad \forall  h>0, \forall \bar{p}\in[0,1] \; .
  \label{condq}
 \end{eqnarray}
 \label{lem:mu}
\end{lem}

The proof is detailed in Appendix \ref{app:proof1}. We can identify the optimal transmission strategy $\mu(p,q)$, according to the following Theorem.
\begin{teo}
 There exists a linear function $g(p) = kp$ such that the optimal transmission strategy $\mu(p,q)$ is to transmit whenever the blockage probability $q$ is lower than $g(p)$, and to defer transmission otherwise.
 \label{teo:optsol}
\end{teo}

\begin{proof}
From (\ref{condq}) of Lemma \ref{lem:mu}, we can state that there exists a function $g(p)$ defined for $0\leq p\leq 1$, such that $\mu(p,q) = 0\, , \, \forall(p,q):q>g(p)$ and that $\mu(p,q) = 1\,, \, \forall(p,q):q<g(p)$. This function splits the square $[0,1]^2$ into two regions: a lower region, where the optimal strategy is to transmit, and an upper region, where the optimal strategy is to defer transmission.
Correspondingly, we can rewrite (\ref{sol_cont}) as
\begin{equation}
 v(p,q) =
 \begin{cases}
  \displaystyle
  p + \gamma(1-q+\gamma^Wq)\int_0^1\int_0^1\phi_{P,Q}(p', q')v(p', q')\de p'\de q' & \text{if } q < g(p)\\
  \displaystyle
  \gamma\int_0^1\int_0^1\phi_{P,Q}(p', q')v(p', q')\de p'\de q' & \text{if } q > g(p) \; .
 \end{cases}
 \label{sol_cont2}
\end{equation}

Note that, from (\ref{condp}), we also infer that $g(p)$ is a non decreasing function of $p$.

The second point we need to highlight is that $v(p,q)$ is a continuous function over all its domain. This is easy to observe if we rewrite (\ref{sol_disc_max}) for the continuous case as
\begin{equation}
 v(p,q) = \max\left(p + \gamma(1-q+\gamma^Wq)C, \gamma C\right) \; ,
 \label{formax}
\end{equation}
Since both $\gamma C$ and $p+\gamma(1-q+\gamma^Wq)C$ are continuous function in $\mathbb{R}^2$, the same holds for their maximum.
With this in mind, we can immediately state that the two branches of (\ref{sol_cont2}) have the same value when $q = g(p)$, and therefore:

\begin{equation}
 p + \gamma(1-g(p)+\gamma^Wg(p))C = \gamma C, \quad \forall p\in[0, 1] \; ,
\end{equation}
which, after some algebraic manipulations, yields
\begin{equation}
 g(p) = \frac{p}{\gamma C(1-\gamma^W)} = kp \; .
 \label{valk}
\end{equation}

We have thus shown that $g(p)$ is a linear function, whose slope $k$ is equal to $(\gamma C(1-\gamma^W))^{-1}$.
\end{proof}

In other words, Thm.~\ref{teo:optsol} states that the admissible risk to be blocked grows linearly with the potential benefit of successfully delivering a data packet, with a slope $k$ depending on the blockage duration $W$.
Using the definition of $k$, we can rewrite the expected reward (\ref{formax}) as
\begin{equation}
 v(p,q) = \gamma C + \max\left(p-\frac{q}{k}, 0\right) \; ,
 \label{funrew}
\end{equation}
which is the fundamental reward equation. For each state $\bm s = (p,q)$, the expected reward is given by the discounted reward of the next time slot, $\gamma C$, plus the difference, if positive, between $p$ and $q/k$, where $p$ is the expected reward if a transmission is attempted. On the other side, $k^{-1} = \gamma C - \gamma^{W+1}C$ can be interpreted as the cost of being blocked, since it is the difference between $\gamma C$, the reward at the next time slot, and $\gamma^{W+1}C$, the reward after $W$ time slots (the length of the blockage). Hence, $q/k$ becomes the cost of being blocked multiplied by the blockage probability or, equivalently, the expected cost of a transmission. A transmission is therefore attempted only if its expected reward $p$ is greater than its expected cost $q/k$, as stated in Thm.~\ref{teo:optsol}.

In order to retrieve an expression for $k$, we start from \eq{funrew} by taking the expectation over $p$ and $q$ on both sides. Recalling that $C$ is by definition equal to $\mathbb{E}_{p,q}[v(p,q)]$, we get:
\begin{equation}
 C = \gamma C + \mathbb{E}_{p,q}\left[\left(p-\frac{q}{k}\right)\chi(q\leq kp)\right] \; ,
\end{equation}
where $\chi(\cdot)$ is the indicator function. Now, replacing $C$ with $(\gamma k(1-\gamma^W))^{-1}$ as per (\ref{valk}) yields, for~\footnote{For $k>1$, the integrals have slightly different expressions, but a similar result is obtained as well.} $k\leq1$
\begin{eqnarray}
 \frac{1-\gamma}{\gamma k(1-\gamma^W)} & = & \int_0^1\int_0^{kp}\left(p-\frac{q}{k}\right)\phi_{P,Q}(p,q)\de q\de p \nonumber \\
 \beta & = & k\int_0^1\int_0^{kp}p\phi_{P,Q}(p,q)\de q\de p - \int_0^1\int_0^{kp}q\phi_{P,Q}(p,q)\de q\de p \; ,
\end{eqnarray}
where
\begin{equation}
 \beta = \frac{1-\gamma}{\gamma(1-\gamma^W)} \; .
 \label{valbeta}
\end{equation}
The only unknown is now $k$, which depends on $\gamma$ and $W$ only through the term $\beta$. 

In the realistic assumption that the interference shows negligible spatial correlation~\cite{Haenggi}, $p$ and $q$ are also uncorrelated, and
in this case, we obtain
\begin{equation}
 \beta = k\int_0^1x\phi_P(x)\Phi_Q(kx)\de x - \int_0^1x\phi_Q(x)\left(1-\Phi_P\left(\frac{x}{k}\right)\right)\de x \; ,
 \label{newbeta}
\end{equation}
where $\Phi_P(p)$ and $\Phi_Q(q)$ are the cumulative distribution functions (cdf) of $p$ and $q$, respectively.

With further algebraic manipulations, for both cases with $k<1$ and $k>1$, we obtain
\begin{equation}
 \beta = \begin{cases}
          k\int_0^1\Phi_Q(kx)\left(1-\Phi_P(x)\right)\de x & \mbox{if } k\leq 1 \\
          \int_0^1\Phi_Q(x)\de x -1 + k\left(1-\int_0^1\Phi_Q(kx)\Phi_P(x)\de x\right) & \mbox{if } k > 1 \; .
         \end{cases}
         \label{solk}
\end{equation}

We observe that \eq{solk} is monotonically increasing with $k$ in its entire domain, so its inversion (to obtain the value of $k$) is possible. In general, this inversion can be done numerically. In some cases a closed form expression for $k$ can be derived, as when both $\{p\}$ and $\{q\}$ are uniformly distributed between 0 and 1. In this case, $\phi_P(p) = 1$ and $\phi_Q(q) = 1$ over all the domain of the function $v(p,q)$,
while $k=1$ gives $\beta = 1/6$. Therefore, inverting \eq{solk} yields
\begin{equation}
 k = \begin{cases}
      \sqrt{6\beta} & \mbox{if } \beta \leq\frac{1}{6} \\
      \beta+\frac{1}{2} + \sqrt{\left(\beta+\frac{1}{2}\right)^2-\frac{1}{3}} & \mbox{if } \beta \geq\frac{1}{6} \; .
     \end{cases}
     \label{sol_uni}
\end{equation}

\section{D2D communications in Cellular Networks: single power level}
\label{sec:D2Dpract_sinpow}
In this section, we apply the theoretical framework developed in Sec.~\ref{sec:tx_policy_1} to a cellular network scenario, in which a D2D communication between two mobile terminals is happening in parallel with a communication from a mobile terminal and the base station, using the same uplink resources.

The licensed user $U$ here is the mobile terminal transmitting to the base station $B$. At the same time, the mobile $S$ is attempting a D2D data transmission to another terminal $D$ in the same band (non-orthogonal transmission), thus potentially interfering with data reception at $B$.

The mobile terminal $S$ can access the channel without the need to coordinate with $B$, avoiding a long latency and allowing the communication even if $S$ is out of the coverage area of $B$. On the other side, $B$ keeps control of the uplink frequency by constantly monitoring its signal to interference and noise ratio (SINR) for the transmission coming from $U$. If the SINR falls under a given threshold, $B$ can block all the interfering communications from $S$, thus preserving the quality of the communication from $U$ for a certain time interval. From the point of view of $B$, there is no need to allocate resources for $S$, neither to schedule its transmissions.

A fixed transmission power level is assumed for $S$. In order to avoid being blocked, $S$ should make a binary choice at each time slot if to transmit or not.
This choice is based on the expected SINR at $D$, which allows the mobile terminal to predict the transmission reliability, and the expected SINR at $B$, which can be used to predict, in case of a transmission attempt, the probability of creating an excessive disturbance at $B$ and consequently being blocked.

The D2D source $S$ has only partial information to estimate these SINRs. Since $U$ is the licensed user, we assume that $S$ can obtain information only on the channels $U-B$ and $U-D$: the base station can forward the information about channel $U-B$ on a proper downlink channel, while the channel $U-D$ can be measured by $D$, and this information can be conveyed to $S$ through an orthogonal, out-of-band control channel.
We will now observe how the optimal strategy derived in Sec.~\ref{sec:tx_policy_1} can be implemented in this scenario.

\subsection{System Model}
We consider a single cell, centered at the BS $B$. At each time slot, a single licensed user is scheduled to transmit on each frequency band. We focus, as a starting point, on a given frequency band, which is assigned to the mobile $U$. We also assume that the idle user $S$ is willing to send data to another idle terminal $D$ via a direct D2D communication.

This scenario is depicted in Fig.~\ref{diag:scenario}, where we highlight also the distances $d_{SB}$, $d_{SD}$, $d_{UD}$ and $d_{UB}$, which correspond to the distances among $S$, $B$, $U$ and $D$. The time is slotted, with slot duration $T$, and slot synchronization is available throughout the network.

\begin{figure}
\label{fig:scenario}
 \begin{center}
  \begin{tikzpicture}[>=stealth',semithick,auto]	 
  \tikzset{basesta/.style = {draw, shape=diamond, text centered, align=center, minimum size=0.8cm, inner sep = 0pt, fill=blue!20},
	   licuser/.style = {draw, shape=circle, text centered, align=center, minimum size=0.6cm, inner sep = 0pt, fill=blue!20},
	   unlicuser/.style = {draw, shape=circle, text centered, align=center, minimum size=0.6cm, inner sep = 0pt},
	   myblo/.style = {draw, shape=circle, text centered, align=center, minimum size=0.5cm, fill=red!20},
	   mytext/.style = {text centered, align=center, minimum size=0.5cm}
		        }
    \node at (0, 0) [basesta](theBS) {\footnotesize $B$};
    \node at (-3, 3) [licuser](licUE) {\footnotesize $U$};
    \node at (-3.5, -0.5) [unlicuser](unS) {\footnotesize $S$};
    \node at (-5, 2.5) [unlicuser](unD) {\footnotesize $D$};
    
    \draw[->] (licUE) to node[auto](step1){$d_{UB}$} (theBS){};
    \draw[->] (unS) to node[auto](step1){$d_{SD}$} (unD){};
    \draw[->,dashed] (unS) to node[auto](step1){$d_{SB}$} (theBS);
    \draw[->,dashed] (licUE) to node[auto](step1){$d_{UD}$} (unD);
    
    \end{tikzpicture}
    \caption{Scenario: D2D communication in a cellular network, with one  licensed user $U$ transmitting to the BS $B$, and one additional user $S$ attempting a D2D transmission to $D$.}    
    \label{diag:scenario}
    \vspace{-1cm}
  \end{center}
\end{figure}
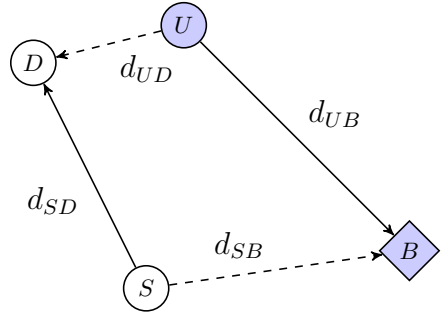

The licensed user $U$ transmits with power $P_U$, whereas $S$ sets its power to $P_S$. The wireless channel between $U$ and $B$ is modeled as a Rayleigh channel, so the SINR at $B$ is
\be
\text{SINR}_B(t) = \frac{A P_U | h_{UB}(t)|^2}{(N_0+I_B)d_{UB}^{\alpha}} \; ,
\label{eq:SINR_UB}
\ee
where $h_{UB}(t)$ is the fading coefficient, modeled as a complex Gaussian random variable with zero mean and unit variance, and assumed to be constant over the entire time slot. $A$ is a fixed path loss term, $N_0$ is the noise power, and $I_B$ is the interference power at $B$ coming from the other transmitters in the network. If a D2D transmission between $S$ and $D$ occurs, the interference term  can be written as $I_B = I_{ic} + I_{SB}$, where $I_{ic}$ is the inter-cell interference, while $I_{SB}$ is the interference due to the D2D transmissions from $S$.

The $\text{SINR}_D(t)$ of the D2D communication can be modeled analogously, by substituting in \eq{eq:SINR_UB} the source $U$ and the destination $B$ with $S$ and $D$, respectively, and $P_U$ with $P_S$. We notice that the interference experienced at $D$ can be written as $I_D = I_{ic} + I_{UD}$, where $I_{ic}$ is the inter-cell interference, while $I_{UD}$ is the interference due to the transmissions from $U$. 

For each transmission, the data packet is correctly received if the corresponding SINR is greater than a threshold $\theta$, according to a decoding threshold model. The probability of successfully receiving a packet from $S$ at $D$ can be written as
\begin{eqnarray}
 p & = & \mathbb{P}\left[SINR_D \geq \theta\right] \nonumber\\
 & = & \mathbb{P}\left[|h_{SD}|^2\geq\frac{(N_0+I_D)\theta d_{UD}^{\alpha}}{AP_S}\right] \nonumber\\
 & = & \exp\left(-\frac{(I_{ic}+N_0)\theta d_{SD}^{\alpha}}{A P_S}\right)\exp\left(-\frac{I_{UD}\theta d_{SD}^{\alpha}}{A P_S}\right)\nonumber \\
 & = & \exp \left( -\frac{\theta}{\gamma_{SD}}\right) \exp \left(-\frac{\theta}{R_D} | h_{UD} |^2\right ) \; ,
\end{eqnarray}
where $\gamma_{SD}$ is the SINR at $D$ without the interference from $U$, as if $I_{UD}=0$, while $R_D$ is simply the ratio between the average power (received by $D$) coming from $S$ and the one coming from $U$. Both $\gamma_{SD}$ and $R_D$ strongly influence the success probability at $D$.

The probability of blockage can also be calculated in a similar way. The cellular base station $B$ will block the D2D transmissions from the terminal $S$ if $\text{SINR}_B$ falls below the decoding threshold $\theta$. As in the case of $I_D$, also the interference experienced at $B$ can be written as sum of two terms $I_B= I_{ic}+I_{SB}$.
The probability of blockage can be expressed as

\begin{eqnarray}
 q & = & \mathbb{P}\left[\text{SINR}_B < \theta\right] \nonumber\\
 & = & \mathbb{P}\left[ AP_U d_{UB}^{-\alpha}|h_{UB}|^2 < \theta (N_0+I_{ic}) + A P_S \theta d_{SB}^{-\alpha} | h_{SB}|^2\right] \nonumber \\
 & = & \exp\left(\frac{N_0+I_{ic}}{A P_S d_{SB}^{-\alpha}} - \frac{P_U d_{UB}^{-\alpha}}{P_S d_{SB}^{-\alpha}}\frac{| h_{UB}|^2}{\theta}\right) \nonumber \\
 & = & \exp\left(\frac{1}{\gamma_{SB}}\right)\exp\left(-\frac{|h_{UB}|^2}{R_B\theta}\right) \; ,
\end{eqnarray}
where $\gamma_{SB}$ is the SINR at $B$ without the interference from $S$, while $R_B$ is the ratio between the average power (received by $B$) coming from $S$ and the one coming from $U$. As in the previous case, $\gamma_{SB}$ and $R_B$ strongly influence the probability of blockage.

We observe that both $p$ and $q$ depend on the values of time varying channels, so they can be described as two independent random processes. At each time slot, they can be represented as two random variables $P$ and $Q$, as a function of the values of 
$\gamma_{UD}$ and $R_D$ (for $p$), and $\gamma_{SB}$ and $R_B$ (for $q$).

The CDF of the success probability $p$ is
\begin{eqnarray}
 \Phi_P(x) & = & \mathbb{P}\left[p\leq x\right]\nonumber \\
 & = & \mathbb{P}\left[-\frac{\theta}{\gamma_{SD}} - \frac{\theta}{R_D} |h_{UD}|^2\leq \ln(x)\right]\nonumber \\
 & = & \mathbb{P}\left[|h_{UD}|^2 \geq -\frac{R_D}{\gamma_{SD}}-\frac{R_D\ln(x)}{\theta}\right] \nonumber \\
 & = & \left\{\begin{array}{ll}
        e^{R_D / \gamma_{SD}} x^{R_D / \theta} & \mbox{if } x \leq e^{-\theta / \gamma_{SD}} \\
        1 & \mbox{if } x > e^{-\theta / \gamma_{SD}}
       \end{array}\right. \; ,
       \label{psuccdist}
\end{eqnarray}
where $\exp(-\theta/\gamma_{SD})$ is the decoding probability with no intra-cell interference. 

Similarly, the CDF of the probability of blockage can be computed as
\begin{eqnarray}
 \Phi_Q(x) & = & \mathbb{P}\left[q < x\right] \nonumber \\
 & = & \mathbb{P}\left[\frac{1}{\gamma_{SB}} - \frac{|h_{UB}|^2}{R_B\theta}\leq \ln(x)\right] \nonumber \\
 & = & \mathbb{P}\left[|h_{UB}|^2 \geq \frac{R_B\theta}{\gamma_{SB}} - R_B\theta\ln(x)\right] \nonumber \\
 & = & e^{-\frac{R_B\theta}{\gamma_{SB}}}x^{R_B\theta} \; .
 \label{ppundist}
\end{eqnarray}

The expressions of these CDFs can be plugged into \eq{solk} in order to calculate the optimal value for $k$, which will give the optimal transmission strategy for $S$ according to Thm.~\ref{teo:optsol}.
In this specific case, we can invert \eq{solk} with a closed form when $k$ is in the interval $[0, e^{\theta/\gamma_{SD}}]$, which corresponds to $\beta \leq \beta_{\ell}$, where
\begin{equation}
 \beta_{\ell} = \beta(k = e^{\theta/\gamma_{SD}}) = \frac{R_D e^{-\theta / \gamma_{UB}}}{\theta(1+R_B\theta)(1+R_B\theta+R_D/\theta)} \; .
 \label{betaell}
\end{equation}
In this interval, in fact, we can write
\begin{equation}
 k = \left(\beta / \beta_{\ell}\right)^{1/ (1+R_b\theta)}e^{\theta / \gamma_{SD}} \; .
 \label{d2d_valk} 
\end{equation}

In the interval $\beta > \beta_{\ell}$, the optimal value of $k$ must be retrieved numerically from 
\be
 \beta = ke^{-\theta/ \gamma_{SD}}(1-z_1) + k^{- R_D / \theta} e^{1 / \gamma_{UD}}(z_1-z_2) + \frac{e^{-\theta / \gamma_{UB}}}{1+R_b \theta} - 1
 \label{d2d_valknum}
\ee
with $\gamma_{UD}$ and $\gamma_{UB}$ defined analogously to $\gamma_{SD}$ and $\gamma_{SB}$, and
\begin{equation}
 z_1 = \frac{1}{1+R_D/ \theta}; \quad \quad z_2 = \frac{e^{-\theta/ \gamma_{UB}}}{1+R_B\theta+R_D/\theta}
 \label{defg1}
\end{equation}

We observe that if $\theta = R_B = R_D = 1$ and if $\gamma_{SD}$ and $\gamma_{SB}$ are high enough, both $p$ and $q$ assume a uniform distribution, thus $k$ can be directly retrieved through \eq{sol_uni}.

\subsection{Strategy implementation: AWA-S}
In this section we propose a practical implementation of an adaptive strategy, named AWA-S, to test the effectiveness of our proposed D2D transmission scheme in the presence of topological information.
In AWA-S, the source $S$ sets the transmission power level to $P_S = \xi d_{SD}^{\alpha}N_0$, where $\xi$ is a predefined target SNR of the D2D communication, coherently with the licensed user $U$, which sets its power level to $P_U = \rho d_{UB}^{\alpha}$, where $\rho$ is the target received power at $B$.

Before starting the transmissions, $S$ should make a preliminary choice between the D2D mode and the D2B mode. 
This choice is made by comparing the expected reward offered by the two modes: if the expected reward for the D2D communication is larger than the D2B ($C_{D2D} \geq C_{D2B}$), then the D2D mode is chosen, otherwise the D2B mode is selected. Notice that in the case of the D2B mode, the uplink channel must be orthogonally shared with $U$.

The calculation of the D2D expected reward follows what detailed in the previous sections. 
We assume that $S$ can collect at each time slot the information about the condition of the channels from $U$, namely the fading coefficients $|h_{UB}|^2$ and $|h_{UD}|^2$. From these two values, the success and the blockage probabilities can be derived at each time slot. Furthermore, the CDFs $\Phi_P(x)$ and $\Phi_Q(x)$ are calculated by using topological information.
The optimal strategy, identified by the value of $k$ obtained from \eq{valk} and \eq{d2d_valk}, can hence be adopted by $S$.
The corresponding expected reward is $C_{D2D} = (k\gamma(1-\gamma^W))^{-1}$, as per \eq{valk}.

In case the D2B mode is chosen,
the two users $U$ and $S$ alternatively transmit to $B$ according to a time division multiple access (TDMA) scheme. 
We assume that i) the transmission towards $B$ is performed with target received power $\rho$ (once every two slots), and ii) $B$ forwards data to $D$ on a separate downlink channel. The expected reward in the D2B mode can be calculated as
\begin{equation}
 C_{D2B} = \sum_{i = 0}^{\infty}e^{-N_0/\rho}\gamma^{2i} = \frac{e^{-N_0/\rho}}{1-\gamma^2} \; ,
\end{equation}
where we do not model the downlink channel, assuming that the bottleneck lies in the uplink.

The AWA-S hence works as follows. If $C_{D2B}>C_{D2D}$, $S$ transmits to $B$ with power equal to $\rho d_{SB}^{\alpha}$, sharing the channel with $U$, and then relies on $B$ to deliver its packets to $D$.
If instead $C_{D2B} \leq C_{D2D}$, the D2D communication is enabled, and $S$ follows the optimal strategy detailed in Thm. \ref{teo:optsol}, with transmit power $\xi d_{SD}^{\alpha}N_0$.

\section{Optimal Transmission Policy: multiple actions choice}
\label{sec:opt_mult_act}
In the previous sections, we have analyzed the case in which the source $S$, in D2D mode, can decide either to transmit with a fixed power level, or to defer the transmission. 

In this section, we extend the analysis by allowing the source $S$ to select also the power level to be used, if a transmission is performed.  
Previously, the action space was binary, now the action should be selected as $a \in \{H,T_1,\dots,T_N\}$, when $N$ different power levels are allowed. The power level for action $T_i$ is $P_S 2^{i-1}$. In this case, 
the success and blockage probabilities are both dependent on the power level $i$ chosen, thus making the optimal action choice a more complex problem.

In order to define an optimal strategy also in this case,
we assume that it is possible to express the success and blockage probabilities 
in terms of two measurable parameters $\pi,\phi\in[0,+\infty)$, independent from the adopted transmission power. We will investigate in Sec.~\ref{sec:mulpowlev} how to compute these parameters in a practical communication case.

In particular, we assume that, given a transmit power level $i$, $p$ and $q$ can be expressed as:
\begin{equation}
 p_i(\pi) = \exp\left(-\frac{a \pi +b}{2^{i-1}}\right) \quad\quad\quad q_i(\phi) = \exp\left(-\frac{c \phi +d}{2^{i-1}}\right)
 \label{defpq}
\end{equation}
where $a$, $b$, $c$ and $d$ are positive constants. In this case, the following Lemma holds.

\begin{lem}
Given an action space with multiple power levels allowed, i.e., $a \in \{H,T_1,\dots,T_N\}$, there exists a linear function $\phi = g_0(\pi) = m \pi + \xi_0$ such that the optimal transmission strategy $\mu(\pi,\phi)$ is to transmit whenever $\phi>g_0(\pi)$ and to defer the transmission otherwise. 
\label{lem:binary}
\end{lem}

From the proof reported in Appendix~\ref{app:proof2}, we find that $g_0(\pi) = a\pi/c + \min_i(\xi(i))$, where $\xi(i) = (2^{i-1}\ln(k^{-1})+b-d)/c$. We also define $h_0(\phi)$ as the inverse of $g_0(\pi)$.
We call $\mathcal{U}$ the set $\{(\pi, \phi):\pi,\phi>0\}$. This set is partitioned into the disjoint sets $\mathcal{A}_0, \mathcal{A}_1,\cdots,\mathcal{A}_N$. $\mathcal{A}_0$ is the subset in which the optimal strategy is to defer transmission, as defined in Lemma~\ref{lem:binary} (that is, $\phi<g_0(\pi)$, or equivalently $\pi>h_0(\phi)$), while in $\mathcal{A}_i$, with $i>0$, the optimal action is $T_i$. With this notation, we can state the following Lemma.

\begin{lem}
For every $1\leq i<N$, the region $\mathcal{A}_i$, with $0<1<N$, can be adjacent only to the regions $\mathcal{A}_{i-1}$ and $\mathcal{A}_{i+1}$. 
\label{lem:many_regions}
\end{lem}

The proof is reported in Appendix~\ref{app:proof3}. Note that $\mathcal{A}_0$ is not included in the Lemma. As to the boundary between two adjacent regions, the following Lemma holds:

\begin{lem}
 The boundary between two existing regions $\mathcal{A}_i$ and $\mathcal{A}_{i+1}$, for $i\geq1$, can always be expressed as a continuous function of either $\pi$ or $\phi$, defined over the entire domain $\mathbb{R}^+$.
 \label{lem:boundaries}
\end{lem}
This Lemma is proved in Appendix~\ref{app:proof4}, where we compute that the boundary between $\mathcal{A}_i$ and $\mathcal{A}_{i+1}$ is given by
\begin{equation}
 g_i^{+}(\pi) = -\frac{d}{c} - \frac{2^i}{c}\ln\left(\frac{1}{2}-\frac{1}{2}\sqrt{1-4kp_{i+1}(\pi)\left(1-p_{i+1}(\pi)\right)}\right) \; ,
\end{equation}
if $k<1$, otherwise it is given by
\begin{equation}
 h_i^{+}(\phi) = -\frac{b}{a} - \frac{2^i}{a}\ln\left(\frac{1}{2}+\frac{1}{2}\sqrt{1-\frac{4}{k}q_{i+1}(\phi)\left(1-q_{i+1}(\phi)\right)}\right) \; .
\end{equation}

We can state the following remarks about the boundary functions $g_i^+(\pi)$ and $h_i^+(\phi)$:
\begin{rem}
 $g_i^+(\pi)\cap g_j^+(\pi)=\emptyset$, and $h_i^+(\phi)\cap h_j^+(\phi)=\emptyset$, $\forall i\neq j$.
 \label{rem:noint}
\end{rem}
In fact, if the intersection between $g_i^+(\pi)$ and $g_j^+(\pi)$ was not empty, there would be a point $(\pi,\phi)$ where the maximum reward can be reached using more than two power levels, which is not possible, due to the characteristics of the reward functions $r(\pi,\phi,i)$ detailed in the proof of Lemma \ref{lem:many_regions}.

\begin{rem}
 When $k<1$, $g_i^+(\pi)\cap\mathcal{U}\neq\emptyset$, $\forall i\in\{1,2,\ldots,N\}$.
 \label{rem:exig}
\end{rem}
This follows from the fact that $g_i^+(\pi)>g_0(\pi)$, $\forall\pi\in\mathbb{R}^+$ (which can be proved through calculations), since $g_0(\pi)$ is an increasing linear function. Note that this implies the existence within $\mathcal{U}$ of all the regions $\mathcal{A}_i$, when $k<1$.
The same does not hold when $k>1$.

\begin{rem}
 For $k<1$, $g_i^+(\pi) > g_j^+(\pi) > g_0(\pi)$, $\forall\pi\in\mathbb{R}^+$, $\forall i>j$. Similarly, for $k>1$, $h_0(\phi) > h_i^+(\phi) > h_j^+(\phi)$, $\forall\phi\in\mathbb{R}^+$, $\forall i>j$.
 \label{rem:order}
\end{rem}
It is not immediate to show the inequality via algebraic derivation. However, one can observe that each function $g_i^+(\pi)$, for $\pi\rightarrow\infty$, approaches asymptotically the linear function $\tilde{g}_i^+(\pi)$:
\begin{equation}
 \tilde{g}_i^+(\pi) = \frac{a}{c}\pi + \frac{1}{c}\left(2^i\ln(k^{-1})+b-d\right)
\end{equation}
When $k<1$, from the fact that $\tilde{g}_i^+(\pi) > \tilde{g}_j^+(\pi)$, $\forall\pi\in\mathbb{R}^+$ for any $i>j$, it follows that $\exists \Pi\in\mathbb{R}^+:\forall\pi>\Pi\,, \, g_i^+(\pi)>g_j^+(\pi)$. Since $g_i^+(\pi)$ and $g_j^+(\pi)$ never intersect, we obtain the statement in Remark \ref{rem:order}.
Proving the same about the functions $h_i^+(\phi)$ is more involved, but can be done by computing the intersections between each $h_i^+(\phi)$ and a properly chosen linear function $\pi = \Pi$, and verifying how these points are sorted.

Using the previous Lemmas and remarks, we can state the following Theorem.
\begin{teo}
Given an action space with $N$ power levels, i.e., $a \in \{H,T_1,\dots,T_N\}$, it is always possible to divide the space $(\pi,\phi)$ into at least 2 and at most $N+1$ continuous regions, such that the optimal policy $\mu(\pi,\phi)$ is always unambiguously defined, with the exception of the boundaries between the regions, which have measure zero.
 \label{teo:optsol_multiple}
\end{teo}

\begin{proof}
The boundary between $\mathcal{A}_i$ and $\mathcal{A}_{i+1}$ is by definition $g_i^+(\pi)$, if $k<1$, or $h_i^+(\phi)$, if $k>1$. Henceforth, these curves, together with $g_0(\pi)$, are the only admissible boundaries between the regions $\mathcal{A}_i$'s.
According to Remark \ref{rem:noint}, these $N$ curves never intersect each other. Since they are continuous functions of either $\pi$ (if $k<1$) or $\phi$ (if $k>1$), it follows that they divide the area $\mathcal{U}$ into at most $N+1$ regions. The number of regions can however be lower. In fact, while it is always $g_i^+(\pi)\cap\mathcal{U}\neq\emptyset$, $\forall i\in\{1,2,\ldots,N\}$, as per Remark \ref{rem:exig}, the same does not hold for the curves $h_i^+(\phi)$.
Indeed, when $k>1$, there can be values of $a$, $b$, $c$, $d$ and $i$ such that $h_i^+(\phi)<0$, $\forall \phi > 0$, meaning that the entire curve lies outside the region $\mathcal{U}$. In this case, no region $\mathcal{A}_i$ exists, as well as no region $\mathcal{A}_j$, with $0<j<i$, due to Remark \ref{rem:order}.
In the extreme case, if $h_N^+(\phi)<0$, $\forall \phi > 0$, then $\mathcal{U}$ is divided in only 2 regions, namely $\mathcal{A}_0$ and $\mathcal{A}_N$, by the curve $g_0(\pi)$ (or, equivalently, $h_0(\phi)$).
\end{proof}

The exact form of the optimal policy $\mu(\pi, \phi)$, for $k<1$, is
\begin{equation}
 \mu(\pi,\phi) = \left\{
	    \begin{array}{ll}
             0 & \text{if } \phi < g_0(\pi) \\
             1 & \text{if } g_0(\pi) < \phi < g_1^+(\pi) \\
             i & \text{if } g_{i-1}^+(\pi) < \phi < g_i^+(\pi)\text{, } \forall i\in\{2,3,\ldots,N-1\} \\
             N & \text{if } \phi > g_{N-1}^+(\pi)
            \end{array}
            \right.
            \label{optpol_gen_lowk}
\end{equation}

The first part of the expression in (\ref{optpol_gen_lowk}) immediately follows from Lemma \ref{lem:binary}.
In the proof of the same lemma, we show that, when $k<1$, the curve $g_0(\pi)$ is the boundary between $\mathcal{A}_0$ and $\mathcal{A}_1$ since $\arg\min(\xi_i) = 1$.
From Remark \ref{rem:exig} we know that all the boundaries $g_i^+(\pi)$ exist within $\mathcal{U}$, and therefore $\mathcal{U}$ is partitioned into $N+1$ areas. Since the functions $g_i^+(\pi)$ are sorted according to Remark \ref{rem:order}, and given Lemmma~\ref{lem:many_regions}, we get that $\mathcal{A}_1$ must lie between $g_0(\pi)$ and $g_1^+(\pi)$.
Analogously, $\mathcal{A}_i$ is bounded by $g_{i-1}^+(\pi)$ and $g_i^+(\pi)$, $\forall i\in{2,3,\ldots,N-1}$. Finally, the region $\mathcal{A}_N$ is the area of $\mathcal{U}$ which lies above $g_{N-1}^+(\pi)$.

Similarly, if $k>1$, the optimal strategy is
\begin{equation}
 \mu(\pi,\phi) = \left\{
 \begin{array}{ll}
  0 & \text{if } \pi > h_0(\phi) \\
  N & \text{if } h_{N-1}^+(\phi) \leq \pi \leq h_0(\phi)\\
  i & \text{if } h_{i-1}^+(\phi) \leq \pi \leq h_i^+(\phi) \text{, } \forall i\in\{2,3,\ldots,N-1\}\\
  1 & \text{if } \pi \leq h_1^+(\phi)
  \label{optpol_gen_highk}
 \end{array}
\right.
\end{equation}
The first part of (\ref{optpol_gen_highk}) follows from Lemma~\ref{lem:binary}, since $h_0(\phi)$ is the inverse function of $g_0(\pi)$, which is a monotonically increasing function.
Now, from Remark \ref{rem:order} it follows that there is a region $\mathcal{R}$ delimited by $h_0(\phi)$ and $h_{N-1}^+(\phi)$. This region must be $\mathcal{A}_N$. In fact, $h_{N-1}^+(\phi)$, according to Lemma~\ref{lem:boundaries}, is the boundary between $\mathcal{A}_{N-1}$ and $\mathcal{A}_N$; however $\mathcal{R}$ cannot be $\mathcal{A}_{N-1}$, since this region must also be delimited by $h_{N-2}^+(\phi)$. It follows that $\mathcal{R}=\mathcal{A}_N$, which proves the second line of (\ref{optpol_gen_highk}).
By exploiting Lemma~\ref{lem:many_regions} for the other boundaries, (\ref{optpol_gen_highk}) is proved,

It must be noted, however, that in this case the existence of a non empty intersection between $\mathcal{U}$ and $h_i^+(\phi)$ is not guaranteed, since it can happen, for some values of $a$, $b$, $c$, $d$ and $i$, that $h_i^+(\phi)<0$, $\forall \phi\in\mathbb{R}^+$. Since the curves $h_i^+(\phi)$ are sorted, according to Remark \ref{rem:order}, we call $i^*$ the maximum $i$ such that $\phi_{i^*}^+(\phi)<0$, $\forall \phi\in\mathbb{R}^+$.
This implies that $\mathcal{U}$ is partitioned into $N+1-i^*$ regions.
If $i^* = N-1$, then the only boundary in $\mathcal{U}$ is $h_0(\phi)$, which splits $\mathcal{U}$ into $\mathcal{A}_0$ and $\mathcal{A}_N$. If on the contrary $i^*< N-1$, then the areas $\mathcal{A}_i$, with $i^*+1\leq i<N$ also exist. In this case, $\mathcal{A}_{i^*+1}$ is the area left of $h_{i^*+1}(\phi)$, while $\mathcal{A}_j$, for $i^*+1<j<N$, is the region between $h_{j-1}^+(\phi)$ and $h_j^+(\phi)$. The optimal strategy $\mu(\pi,\phi)$ can still be expressed as in (\ref{optpol_gen_highk}), but it may happen that no couple $(\pi,\phi)\in\mathcal{U}$ can satisfy the conditions to get $\mu(\pi,\phi) = i$, for some $0<i<N$.

\section{D2D communications in Cellular Networks: multiple power levels}
\label{sec:mulpowlev}
In this section, we assume that $N$ power levels $P_1,P_2,\ldots,P_N$ are allowed, with $P_i = P_S2^{i-1}$, and we apply the results obtained in Sec.~\ref{sec:opt_mult_act} in a D2D network scenario. In this case, both the success and the blockage probabilities depend on the adopted transmission power.

The success probability is given by the probability that the SINR at $D$ is greater than $\theta$, while the blockage probability is the probability that the SINR at the base station $B$ falls below $\theta$. For each trasmit power level $i$, we can hence express them as:
\be
 p_i(\pi) = \exp\left(-\frac{\theta(\pi+N_0)}{2^{i-1}AP_S d_{SD}^{-\alpha}}\right) \; ,
 \ee
 \be
 q_i(\phi) =
 \begin{cases}
  \displaystyle\exp\left(-\frac{\phi-\theta N_0}{2^{i-1}\theta AP_S d_{SB}^{-\alpha}}\right) & \text{if } \phi > \theta N_0 \\
  1 & \text{if } \phi \leq \theta N_0 \; ,
 \end{cases}
\ee
where $\pi = AP_U d_{UD}^{-\alpha}|h_{UD}(t)|^2$ is defined as the interference at $D$ caused by the licensed user $U$, while $\phi = AP_U d_{UB}^{-\alpha}|h_{UB}(t)|^2$ is the useful signal from $U$ at $B$.

The expressions of $p_i(\pi)$ and $q_i(\phi)$ are thus analogous to those in (\ref{defpq}), with $a=\theta/(AP_S d_{SD}^{-\alpha})$, $b = N_0a$, $c=1/(\theta AP_S d_{SB}^{-\alpha})$, and $d = -\theta N_0 c$. 
The only issue lies in the fact that $d<0$. However, it can be shown that even in this case, Theorem \ref{teo:optsol_multiple} is still valid, with $g_0(\pi)$ replaced by 
\begin{equation}
 \tilde{g}_0(\pi) = g_0(\pi)\mathbb{S}\left(\pi + \frac{2^{i-1}}{a}\ln(k^{-1}) + \frac{b}{a}\right) \; ,
\end{equation}
where $\mathbb{S}(\cdot)$ is the Heaviside step function. Notice that, since $\pi,a,b>0$, when $k<1$ we have again $\tilde{g}_0(\pi) = g_0(\pi)$.

The optimal strategy can be written as in (\ref{optpol_gen_lowk}), if $k<1$, and as in (\ref{optpol_gen_highk}) if $k>1$, 
by substituting 
$g_0(\pi)$ with $\tilde{g}_0(\pi)$, and $h_0(\phi)$ with $\tilde{h}_0(\phi) = \max(h_0(\phi), 2^{i-1}\ln(k)/a - b/a)$.
\begin{figure}
    \centering
    \subfigure[$k = 0.5571$]
    {\includegraphics[width=\figw]{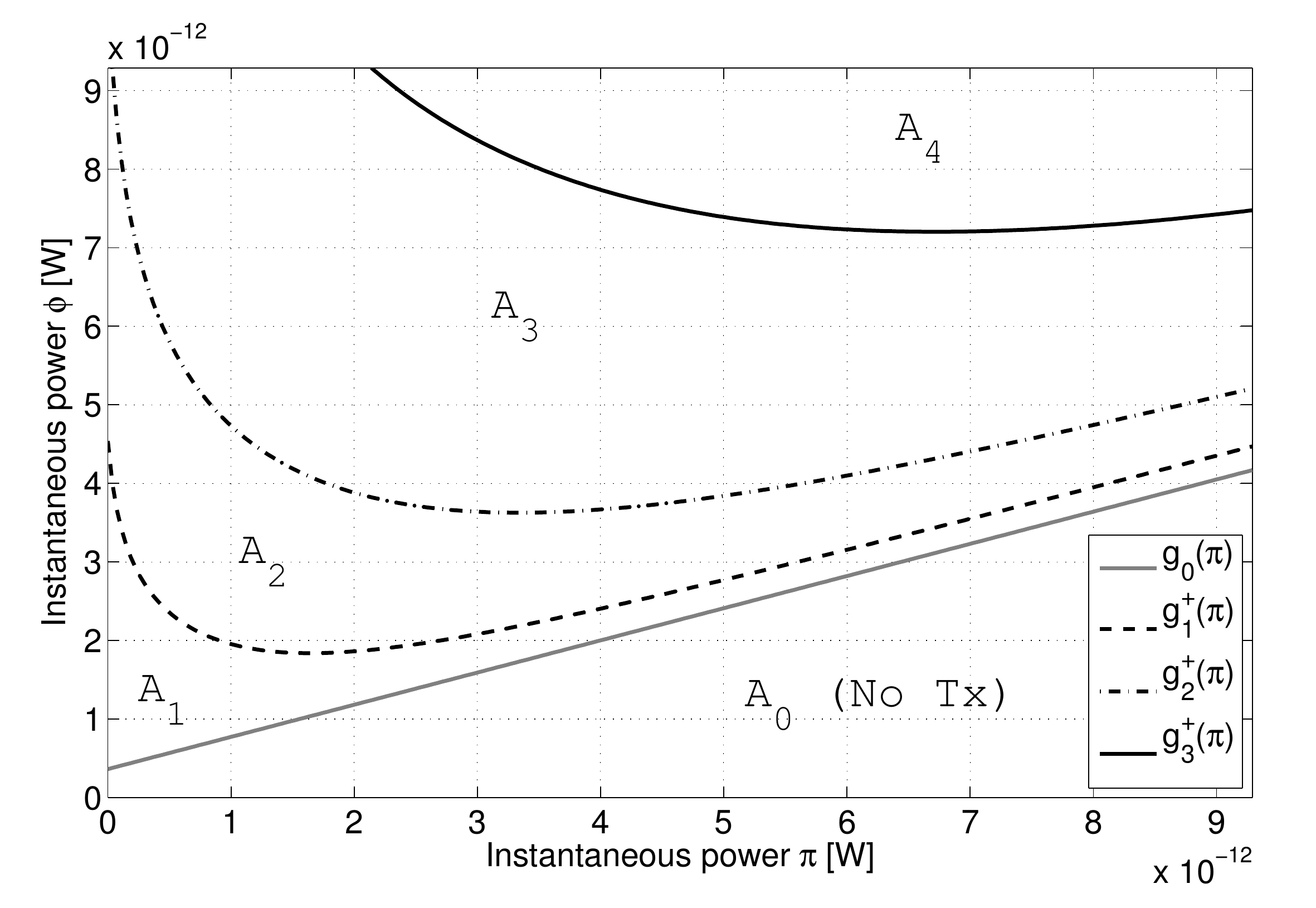}
     \label{fig:opt_multi_klow}
     }
     \subfigure[$k = 1.1613$]
     {\includegraphics[width=\figw]{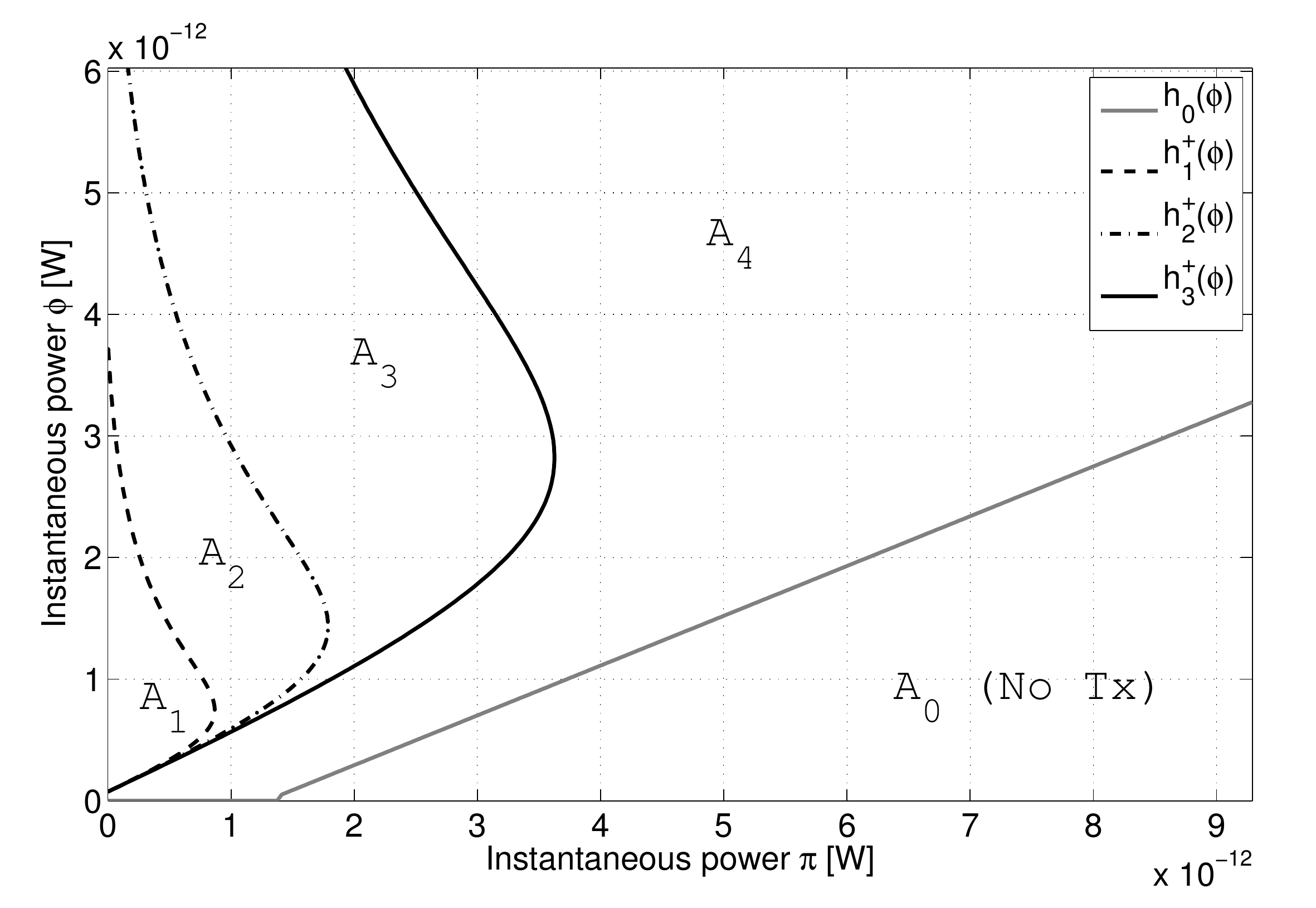}
     \label{fig:opt_multi_khigh}
     }
     \caption{\small The optimal policy for the scenario with the following coordinates: $B = (0,0)$, $S = (100, 0)$, $D = (100, 80)$ and $U = (0, 120)$. There are 4 power levels $P_i = 0.05\times2^{i-1}\spa W$, with $i\in\{1,2,3,4\}$. In case (a) we have $W = 10$, whereas in (b) $W = 3$, resulting in $k = 0.5571$ and $k = 1.1613$, respectively.} 
     \label{fig:opt_multi}
     \vspace{-1cm}
\end{figure}
An example of the shape of the optimal policy $\mu(\pi,\phi)$ is reported for $k<1$ in Fig.~\ref{fig:opt_multi}-(a), and for $k>1$ in Fig.~\ref{fig:opt_multi}-(b).

\subsection{Strategy implementation with multiple power levels: AWAm-S}
In Sec.~\ref{sec:D2Dpract_sinpow}, we introduced the basic implementation of AWA-S.
In this section, we provide an extension of this context aware strategy, namely AWAm-S.
The strategy behavior follows the one described for AWA-S, but now multiple predefined power levels are available for D2D communications. As before, the choice between the D2D and the D2B mode is made by comparing the expected reward $C_{D2D}$ and $C_{D2B}$. In this case, however, $C_{D2D}$ is derived as per the optimal policy for multiple power level scenario described in Sec.~\ref{sec:mulpowlev} and numerically computed. An example of an algorithm to derive it is reported in Appendix~\ref{app:algo}.

If D2D mode is selected, the source terminal $S$ uses the value of $k$ (previously derived to compute $C_{D2D}$) which, together with topological information, fully defines the curves $g_i^+(\pi)$ and $h_i^+(\phi)$, and therefore the optimal strategy. Then, at each time slot it collects information about the power received from $U$ at $D$ and $B$, corresponding to $\pi$ and $\phi$, respectively, and acts as per the optimal strategy described in (\ref{optpol_gen_lowk}) or (\ref{optpol_gen_highk}).

\section{Results}
\label{sec:results}
In this section, we present the performance of the AWA-S and the AWAm-S strategies. These strategies are compared to a state-of-the-art strategy~\cite{modpowcon}, which is renamed as GEO-S in this paper, and briefly described in Sec.~\ref{sec:geoS}. 

\subsection{Simulation scenario}
We focus on the single channel scenario depicted in Fig.~\ref{diag:scenario}, with one node $U$ transmitting to the BS $B$, and one additional node $S$ attempting a transmission to node $D$. In this scenario, node $S$ has the choice to transmit in D2D mode to $D$, or to rely on the BS to forward the packet to $D$. 

In the D2D case, $S$ should make sure to limit its interference to $B$, since the BS $B$ constantly monitors the SINR of the transmissions from $U$. Whenever this SINR falls below the decoding threshold $\theta$, while a simultaneous D2D transmission from $S$ has been performed, the BS forces $S$ to remain silent for a predefined amount of time, equal to $W$ time slots.
In the D2B case, instead, $S$ will share the channel (in TDMA) with $U$, so both users will be allowed to transmit for one time slot every two.

We do not model the downlink channel between $B$ and $D$, since we assume that the bottleneck lies in the uplink channel. Moreover, since different frequency bands are utilized for uplink and downlink, we consider full-duplex relaying at the BS.

We set the cellular radius to $R=250\spa\rm{m}$, the maximum distance between $S$ and $D$ to $L=100\spa\rm{m}$, the path loss exponent to $\alpha = 4$, the fixed path loss term to $A=1$, the target SNR at $B$ to $\rho= N_0 = -90\spa\rm{dBm}$, the decoding threshold to $\theta = 0\spa\rm{dB}$, and the discount factor of the MDP to $\gamma=0.99$. 

The results are obtained by averaging over $5\cdot 10^6$ time slots, obtained from $5\cdot 10^3$ randomly generated topologies. In each topology, we randomly deployed user $U$ within the cell ($d_{UB} \leq R$), user $S$ in the inner part of the cell ($d_{SB} \leq 0.75\cdot R$), and user $D$ in the same inner part, within a certain distance from $S$ ($d_{DB} \leq 0.75\cdot R$ and $d_{SD} \leq L$). 

\subsection{GEO-S}
\label{sec:geoS}
The GEO-S is a state-of-the-art strategy~\cite{modpowcon} based only on geographic considerations. The idea is to let a D2D communication between $S$ and $D$ as long as this does not cause an excessive expected interference at the BS $B$, otherwise $S$ switches to the D2B mode and alternates its transmissions with $U$.

The expected interference is estimated based only on geographic considerations, and the D2D mode is chosen if 
\begin{equation}
 T_d d_{SD}^{-\alpha} > d_{SB}^{-\alpha} \; ,
\end{equation}
where $T_d\geq 0$ is a tunable parameter, which is set to $T_d = 0.8$ in our simulations. Otherwise, the communication is done in D2B mode. 

In the case of the D2D mode is chosen, $S$ transmits with power equal to $\rho d_{SD}^{\alpha}$, otherwise $U$ and $S$ alternatively transmit to $B$ with power $\rho d_{UB}^{\alpha}$ and $\rho d_{SB}^{\alpha}$, respectively.
 
\subsection{Tradeoff of AWA-S}

As explained above, the target SNR at $B$ from user $U$ is fixed, and equal to $\rho/N_0$.
In order to protect the transmissions from $U$, the BS $B$ can properly set the length $W$ of the blockage period in the case of disturbance from $S$. By increasing the value of $W$, $B$ can limit the interference from $S$, thus reducing the impact of the D2D communication. However, if $W$ becomes too large, the choice of the D2D mode will not be advantageous for $S$ any more, i.e., $C_{D2D} < C_{D2B}$, so $S$ will switch to the D2B mode, and in this case $U$ will be able to transmit only in one time slot out of two, while the others will be allocated to $S$. 

On the other side, in AWA-S the user $S$ can set its target SNR $\xi$ at $D$. In general, setting a higher value for $\xi$ increases the received SINR at $D$, but it also implies an increased risk of triggering a blockage period from $B$.

In the following, we investigate the tradeoff between these two parameters ($\xi$ and $W$) and the performance of $U$ and $S$ in terms of throughput, defined as the number of packets received at the destination per time slot. Each source can attempt the transmission of a single packet in a time slot, thus the maximum throughput achievable, with a perfect channel and without interference, is equal to $1$.

First, we show how $\Omega_U$ and $\Omega_S$ (the throughput of $U$ and the one of $S$) vary with $W$ and $\xi$, then we show the throughput of the whole system ($\Omega_{U+S} = \Omega_U+\Omega_S$), and finally we focus on the fairness of the system, showing the minimum throughput between $U $ and $S$ ($\Omega_{\text{min}} = \min \{ \Omega_U,\Omega_S\}$).

In the baseline scenario, where D2D communications are not allowed, both $U$ and $S$ access the uplink channel for $50\%$ of the time, and the target SNR for both is set equal to $\rho/N_0$. The average throughput for each of them is
\begin{equation}
 \Omega_{0} = 0.5 \; e^{-\theta N_0 / \rho} = 0.18 \; \text{pkt/slot} .
\end{equation}

\subsection{Performance comparison}
\begin{figure}
    \centering
    \includegraphics[width=\figw]{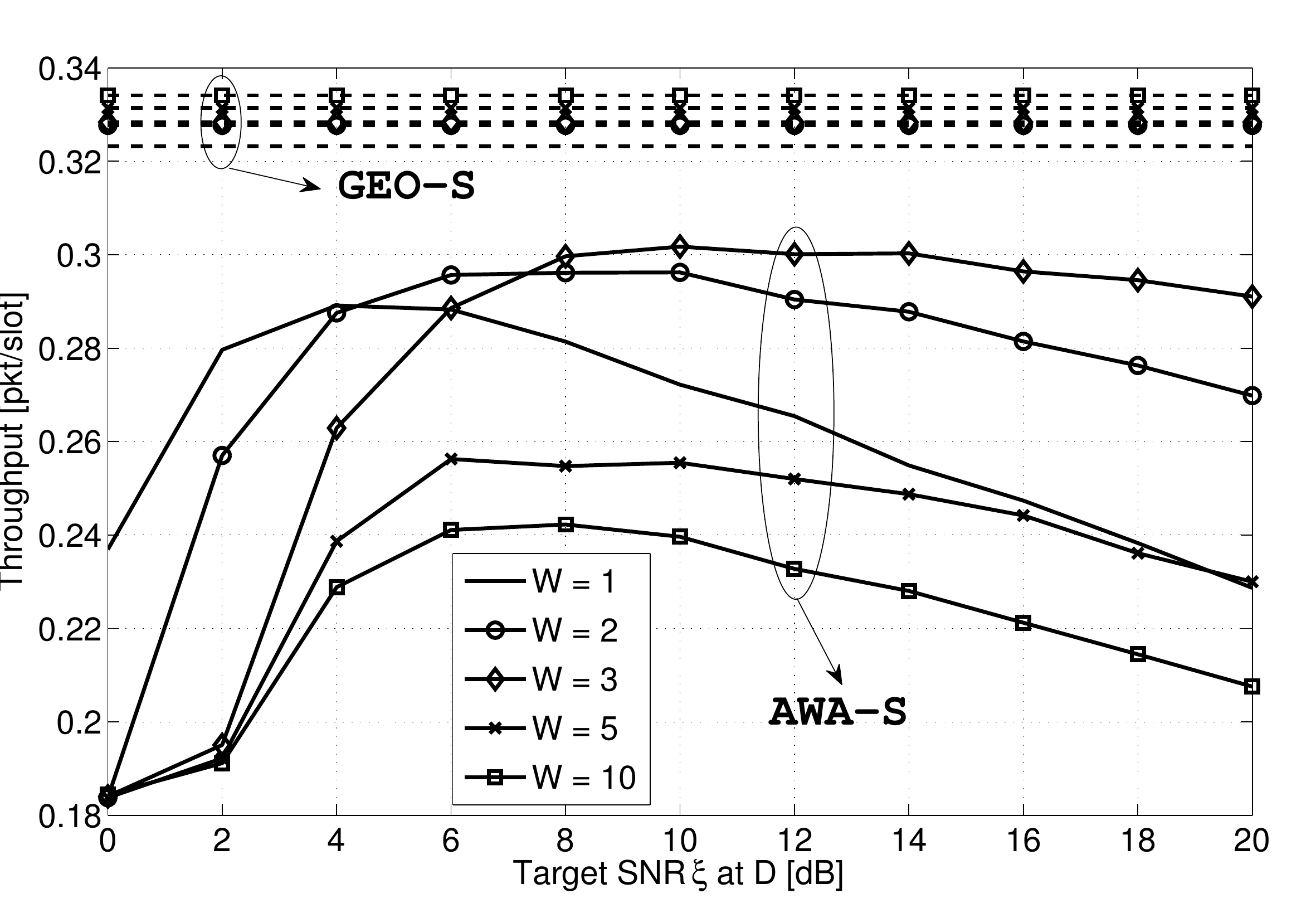}
     \caption{\small The average throughput of $U$, as a function of the target SNR $\xi$ that is set by $S$.}
     \label{fig:throU_vs_mu}
     \vspace{-1cm}
\end{figure}

In Fig.~\ref{fig:throU_vs_mu}, we depict $\Omega_U$, the throughput from U, as a function of the D2D target SNR $\xi$ and for different values of $W$. For AWA-S, we observe that for each value of $W$ there is a value of $\xi$ for which $\Omega_U$ is maximized. For higher values of $\xi$, the interference from $S$ becomes more significant. For lower values of $\xi$ instead $S$ is more likely to opt for a D2B transmission, and consequently $U$ can use only half of the time slots to transmit. 

On the other side, if the value of $\xi$ has been set by $S$, there exists an optimal value of $W$ to maximize $\Omega_U$.
A too low value of $W$ makes it convenient for $S$ to transmit even in case it will disturb $U$, while a too high value of $W$ will make it convenient for $S$ to switch to a D2B mode, thus exclusively  using half of the resources. We observe also that the optimal value for $W$ increases as $\xi$ increases, since a higher $\xi$ means a higher power for $S$, which also means a higher disturbance for $U$.
Finally, we observe that $\Omega_U$ is always higher for the GEO-S strategy. This is not surprising, since the GEO-S is designed to protect the transmissions by $U$ from any interference coming from $S$.
 
\begin{figure}
    \centering
    \includegraphics[width=\figw]{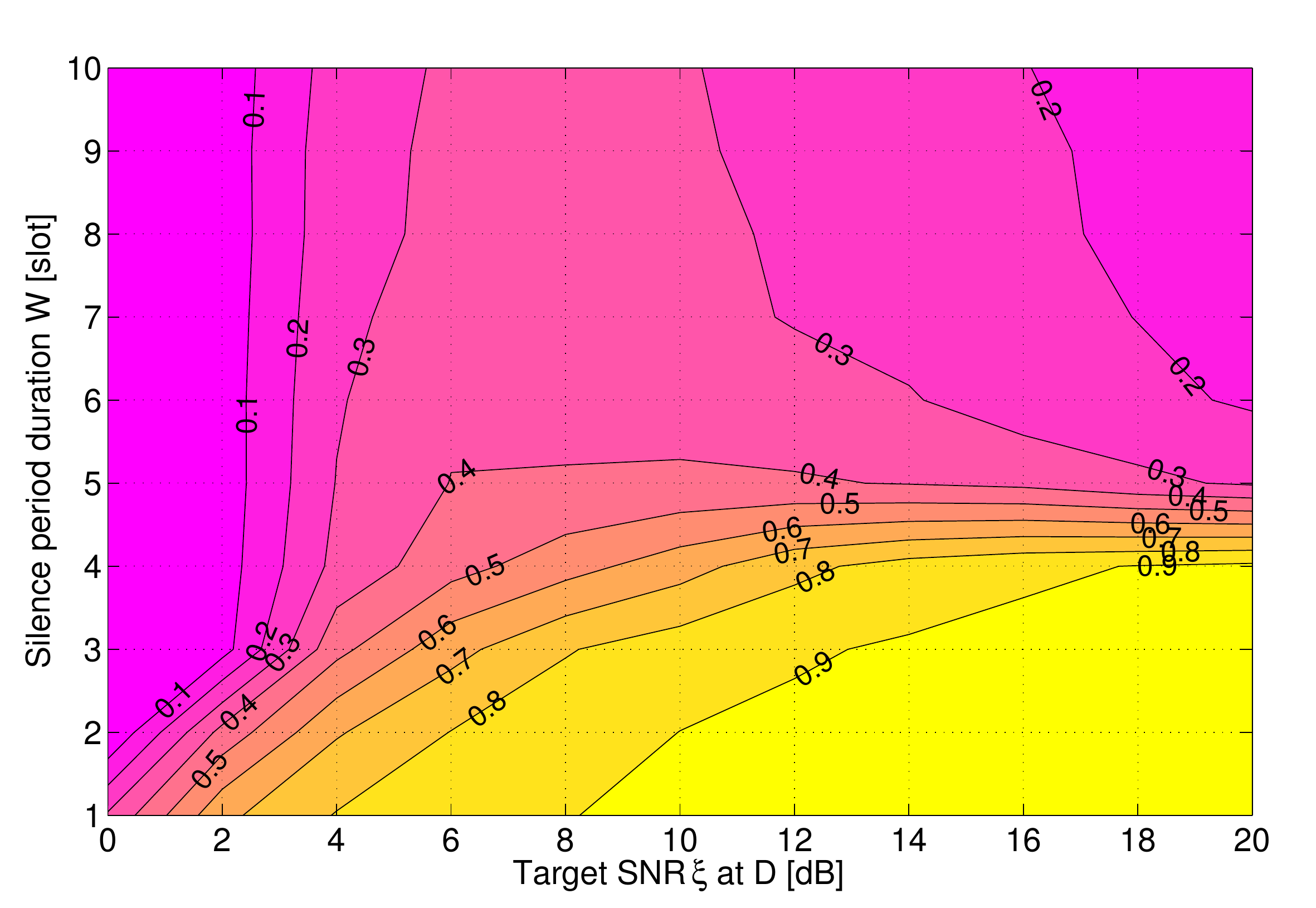}
     \caption{\small Probability of $S$ choosing the D2D transmission mode, as a function of both $W$ and $\xi$.}
     \label{fig:AwaMode}
     \vspace{-1cm}
\end{figure}

In order to better understand the functioning of AWA-S, in Fig.~\ref{fig:AwaMode} we plot the probability of choosing the D2D mode for $S$ in this scenario, as a function of the values of $W$ and $\xi$. In general, we observe that if $\xi$ is fixed, the probability of selecting the D2D mode decreases with $W$, as expected since a higher value of $W$ means a longer blockage period. 
We also observe that for $W\leq4$, the probability of choosing the D2D mode increases with $\xi$, since in case of a short blockage period, it is more convenient for $S$ to transmit in D2D at higher power. This observation is indeed no longer true for longer blockage periods, i.e., when $W \geq 4$.

\begin{figure}
    \centering
    {\includegraphics[width=\figw]{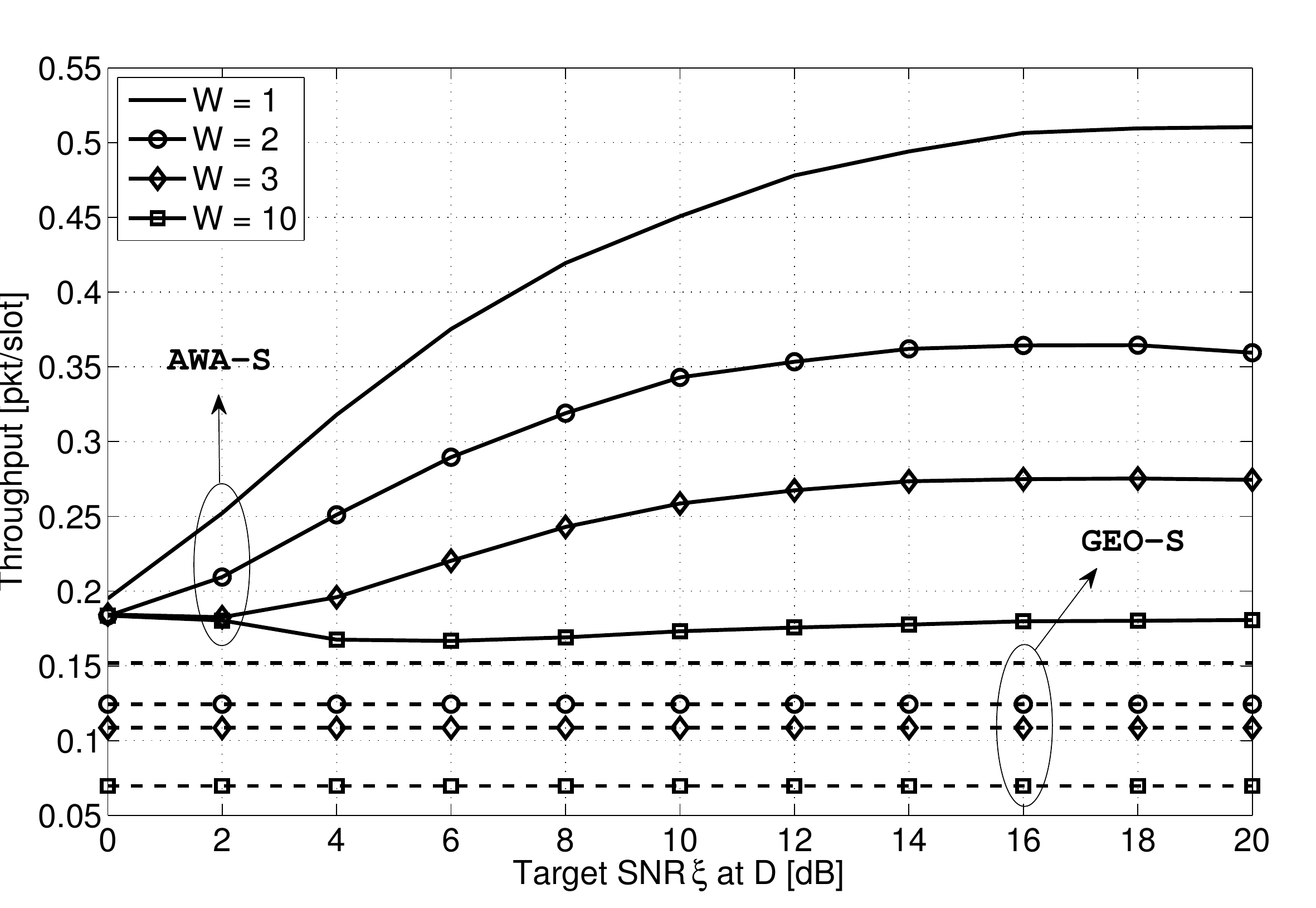}
     \caption{\small The average throughput of $S$, as a function of the target SNR $\xi$ that is set by $S$.}
     \label{fig:throS_vs_mu}
     }
     \vspace{-0.5cm}
\end{figure}

The value of $\Omega_S$, the throughput of $S$, is shown in Fig.~\ref{fig:throS_vs_mu}.
We observe that, in the AWA-S case, it is convenient for $S$ to increase the value of $\xi$, at least for low values of the blockage duration, $W \leq 4$. In any case, the value of $\xi$ can not be arbitrarily increased, otherwise a blockage will happen after every attempt of transmitting in D2D.
We also notice that $\Omega_S$ is much higher in the case of AWA-S than GEO-S, for any of the considered values of $W$ and $\xi$.

\begin{figure}
    \centering
     {\includegraphics[width=\figw]{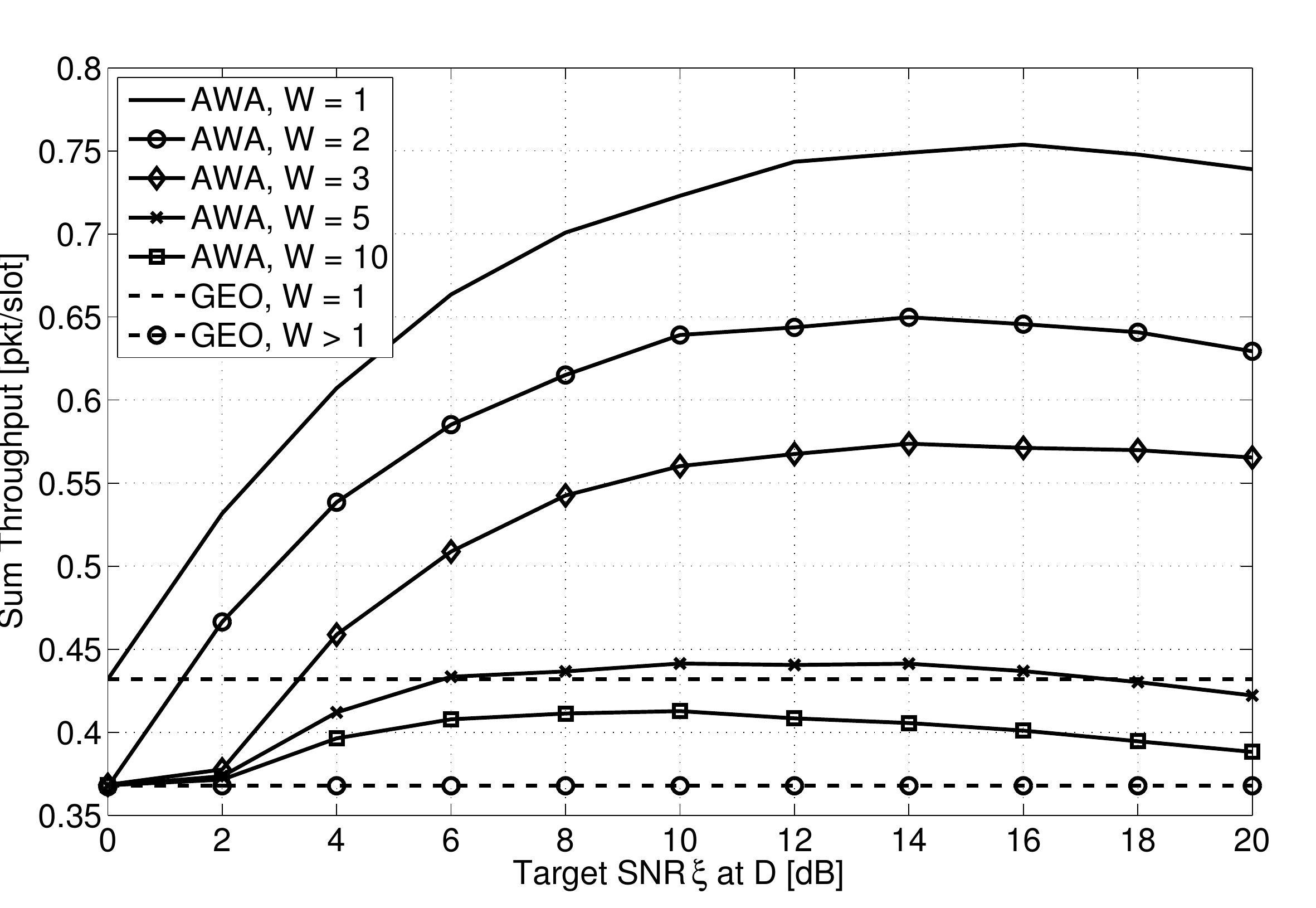}
     }
     \caption{\small The sum throughput of $U$ and $S$, as a function of the target SNR $\xi$ that is set by $S$.}
     \label{fig:throAll_vs_mu}
     \vspace{-1cm}
\end{figure}

We have seen that in general AWA-S guarantees a higher performance than GEO-S for user $S$, while the opposite is true for user $U$. In order to analyze the cost-benefit balance for AWA-S, we depict in Fig.~\ref{fig:throAll_vs_mu} the total system throughput, $\Omega_{U+S}$. In general, the maximum for the system throughput $\Omega_{U+S}$ is obtained for $W=1$, and for each value of $W$ the performance are maximized for one finite value of $\xi$. 
The total throughput increase for AWA-S as compared to GEO-S is particularly significant for low values of $W$ and $\xi >4\spa\rm{dB}$.
In particular, for $W=1$ and $\xi =16\spa\rm{dB}$, AWA-S outperforms GEO-S by about $75\%$ in terms of total throughput.

\begin{figure}
    \centering
     {\includegraphics[width=\figw]{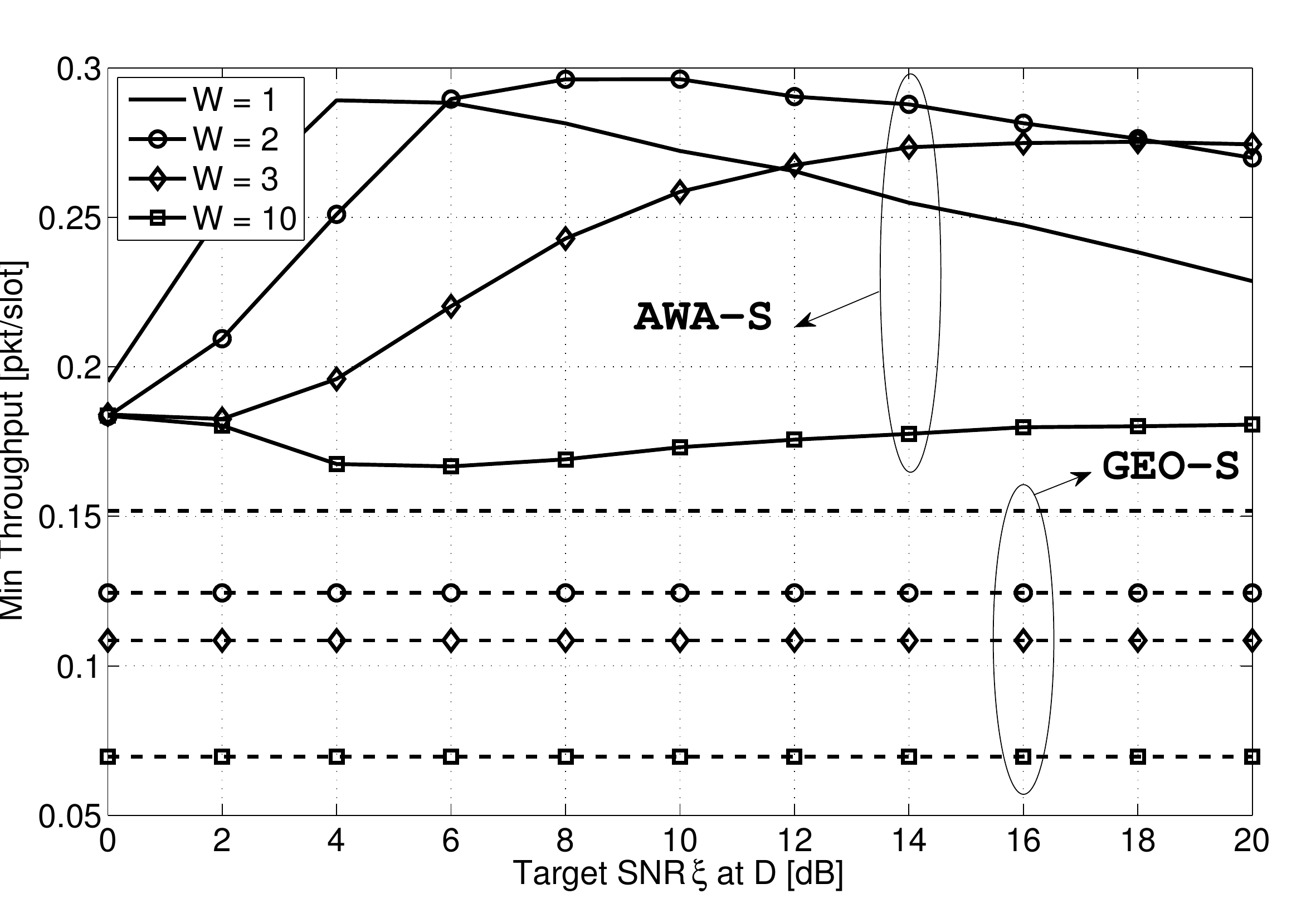}
     }
     \caption{\small The minimum throughput of $U$ and $S$, as a function of the target SNR $\xi$ that is set by $S$.}
     \label{fig:throMin_vs_mu}
     \vspace{-0.5cm}
\end{figure}

This significant increase in $\Omega_{U+S}$ comes at the cost of a decrease in $\Omega_{U}$, the throughput from $U$. In order to compare the fairness of the two strategies, in Fig.~\ref{fig:throMin_vs_mu} we depict the value of the minimum throughput between $\Omega_U$ and $\Omega_S$, i.e., $\Omega_{\text{min}}$. 
We observe that $\Omega_{\text{min}}$ is much higher for AWA-S than for GEO-S, since AWA-S can better balance the available resources. As shown in the figure, the value of $\xi$ should be limited, in order not to impair the transmissions from $U$.
The highest minimum throughput for AWA-S is obtained for $\xi = 10\spa\rm{dB}$ and $W = 2$. With these values, $\Omega_U= 0.30\spa\rm{pkt/slot}$, while $\Omega_S = 0.34\spa\rm{pkt/slot}$, and the total throughput is $\Omega_{U+S} = 0.64\spa\rm{pkt/slot}$. By using GEO-S in the same scenario, with $W=2$, we obtain a higher $\Omega_U = 0.33\spa\rm{pkt/slot}$, but only $\Omega_S = 0.12\spa\rm{pkt/slot}$, thus $\Omega_{U+S} = 0.45 \spa\rm{pkt/slot}$. In the same scenario, if D2D can not be employed and $U$ and $S$ will just alternate in transmitting to $B$, we obtain $\Omega_{U+S} = 2 \Omega_0 = 0.37\spa\rm{pkt/slot}$).

In other words, we can set AWA-S to achieve the maximum fairness. Even in this case, the relative gain in terms of system throughput over GEO-S is of about $42\%$, while over the case of no D2D transmission the relative gain is about $73\%$.

\subsection{Multiple power levels: AWAm-S}

\begin{figure}
    \centering
    {\includegraphics[width=\figw]{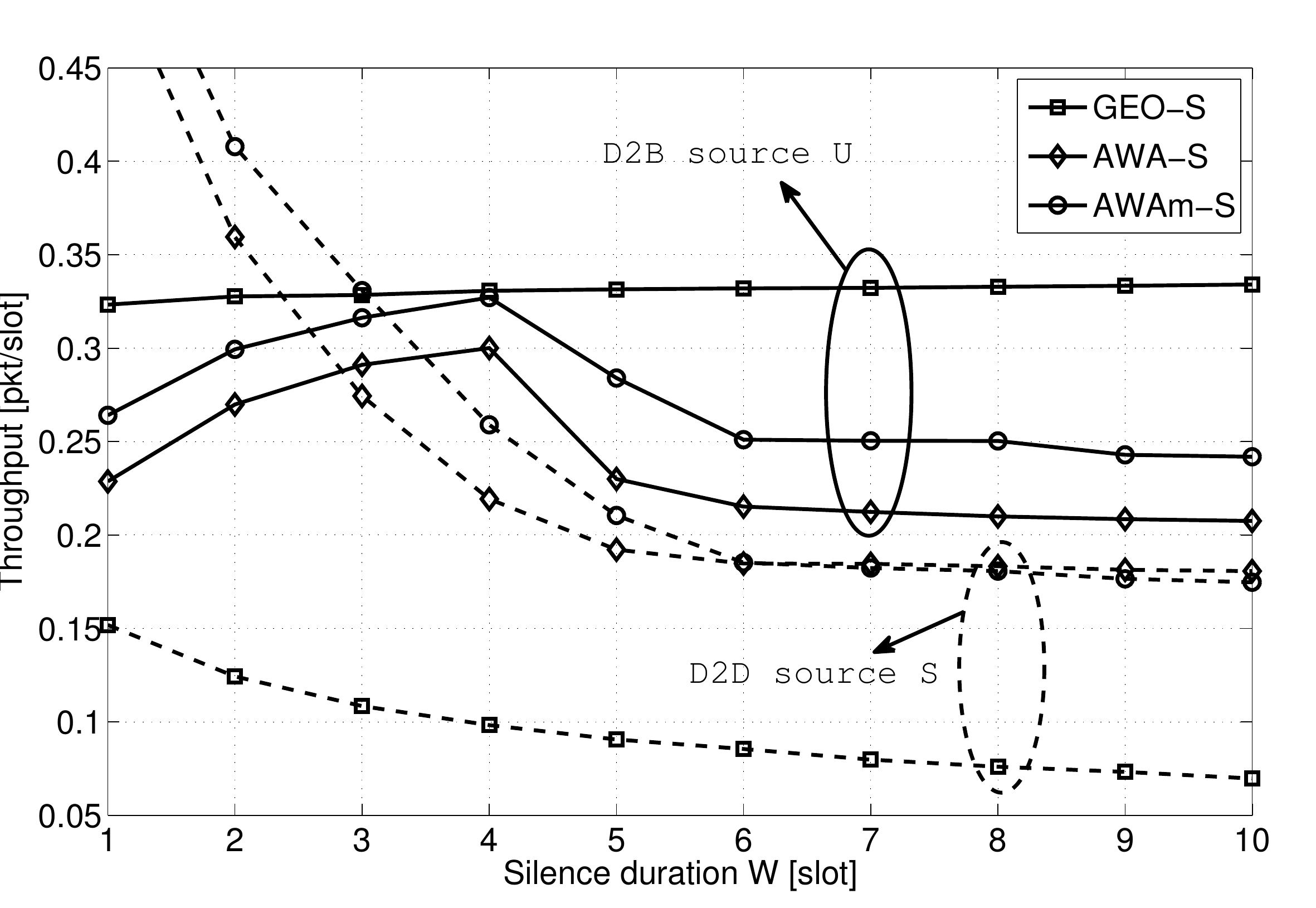}
     \caption{\small The throughput of $U$ and $S$, as a function of the blockage duration $W$, for GEO-S, AWA-S ($\xi=20\spa\rm{dB}$) and AWAm-S.}
     \label{fig:Multi_ThroSeU_vs_W}
     }
     \vspace{-1cm}
\end{figure}

In Fig.~\ref{fig:Multi_ThroSeU_vs_W},  we plot $\Omega_U$ and $\Omega_S$ for the GEO-S, the AWA-S (with $\xi = 20\spa\rm{dB}$), and the AWAm-S (with $12$ power levels from $-13\spa\rm{dBm}$ to $20\spa\rm{dBm}$).
Regarding $\Omega_U$, AWAm-S outperforms AWA-S for all the values of $W$ considered. It is interesting to focus on the performance in $W=4$, where AWAm-S obtains the same value for $\Omega_U$ as GEO-S. In the same point indeed AWAm-S is able to outperform GEO-S of about $160\%$ in terms of the relative throughput $\Omega_S$. In general, AWAm-S outperforms AWA-S for all the throughput values considered. 

The total throughput $\Omega_{U+S}$ can be easily derived from the figure. We observe that the maximum improvement with respect to GEO-S in the total throughput comes for $W=1$, at the cost of a severe decreasing of the performance for $U$. If we consider fairness instead, i.e., $\Omega_{\text{min}}$, the optimal solution for AWAm-S is obtained for $W = 3$. With this value, AWAm-S attains a minimum throughput of $\Omega_{\text{min}} = 0.32\spa\rm{pkt/slot}$, which is about 15\% higher than AWA-S, and almost three times the one offered by GEO-S.

\section{Conclusions}
\label{sec:conclu}
In this paper, we theoretically analyzed how context awareness can be used to exploit D2D communications with limited impact in terms of interference on the other ongoing communications.
We derived two theorems to define an optimal strategy that aims at setting up D2D communications only when the gain in terms of throughput overcomes the cost in terms of interference impact.
By comparing our strategy with the state-of-the-art work in the same scenario, we showed that a distributed context-aware D2D communications scheduling can lead to a substantial gain in terms of sum throughput and fairness.
In a future work, we plan to exted the analysis to more complex scenarios, where inter-cell interference is also modeled, and multi-hop communications can be established as well.

\appendices

\section{Proof of Lemma~\ref{lem:mu}.}
\label{app:proof1}
\begin{proof}
 We show here the proof for (\ref{condq}), an analogous one can be sketched for (\ref{condp}). Since $\mu(p,q)$ is a binary function, proving (\ref{condq}) is equivalent to prove the two following implications:
 \begin{eqnarray}
  \mu(\bar{p},q) = 0 & \Rightarrow & \mu(\bar{p}, q+h) = 0;\\
  \mu(\bar{p},q) = 1 & \Rightarrow & \mu(\bar{p}, q-h) = 1;
 \end{eqnarray}
For the first one, recall that $\mu(\bar{p},q) = 0$ implies that the expected reward obtained by avoiding transmission is greater than that obtained through a packet transmission. Looking at (\ref{sol_cont}), this means that by hypothesis it must be $\gamma C > \bar{p} + \gamma(1-q+\gamma^Wq)C$, being $C$ the value of the double integral.

Now, if it were $\mu(\bar{p}, q+h) = 1$, meaning that the optimal strategy at $(\bar{p}, q+h)$ is to transmit, it would be:
\begin{eqnarray}
 \gamma C & < & \bar{p} + \gamma(1-q-h+\gamma^W(q+h))C \nonumber\\
 & < & \bar{p} + \gamma(1-q+\gamma^Wq)C - \gamma hC(1-\gamma^W) \nonumber \\
 & < & \bar{p} + \gamma(1-q+\gamma^Wq)C
\end{eqnarray}
which contradicts the hypothesis.

Similarly, for the second implication, assuming $\mu(\bar{p},q) = 1$ is equivalent to the hypothesis $\gamma C < \bar{p} + \gamma(1-q+\gamma^Wq)C$. If it were $\mu(\bar{p}, q-h) = 0$, this would mean:
\begin{eqnarray}
 \gamma C & > & \bar{p} + \gamma(1-q+h+\gamma^W(q-h))C \nonumber\\
 & > & \bar{p} + \gamma(1-q+\gamma^Wq)C + \gamma hC(1-\gamma^W) \nonumber \\
 & > & \bar{p} + \gamma(1-q+\gamma^Wq)C
\end{eqnarray}
which is again contradictory. This proves (\ref{condq}), and a similar argument proves (\ref{condp}).
\end{proof}

\section{Proof of Lemma~\ref{lem:binary}}
\label{app:proof2}
\begin{proof}
Any transmission strategy has an expected reward $C$, and a punishment cost defined as $\gamma C(1-\gamma^W)$. For the optimal strategy, according to the previous case, we call $k$ the reciprocal of this cost. Recalling the fundamental reward equation (\ref{funrew}), it is preferable to defer transmission, rather than transmitting, when $q > kp$.
When multiple power levels are available, action $W$ is selected if it is preferable to the transmission with any power level. Therefore, it must be
\be
q(i) > kp(i),\quad \forall i\in(1,2,\ldots,N)
\ee
Expressing the probabilities as a function of the parameters $\pi$ and $\phi$, as stated above, turns the condition into:
\be
 \phi < \frac{a}{c}\pi + \frac{1}{c}\left(2^{i-1}\ln(k^{-1})+b-d\right) = \frac{a}{c}\pi + \xi(i),\quad \forall i\in(1,2,\ldots,N)
\ee
Hence, action $W$ is the optimal one for:
\be
 \phi < \frac{a}{c}\pi + \min_i(\xi(i))
\ee
This proves the lemma, with $m = a/c$, while $\xi_0 = \xi(1)$, if $k<1$, and $\xi_0 = \xi(N)$ otherwise.
\end{proof}

\section{Proof of Lemma \ref{lem:many_regions}}
\label{app:proof3}
\begin{proof}
We recall that the expected reward at any point $(\pi, \phi)\in\mathcal{U}$ using power level $i$ can be expressed as:
\begin{equation}
 r(\pi, \phi, i) = \gamma C + \exp\left(-\frac{a\pi+b}{2^{i-1}}\right) - \frac{1}{k}\exp\left(-\frac{c\phi+d}{2^{i-1}}\right)
 \label{reweq}
\end{equation}
where we plugged (\ref{defpq}) into (\ref{funrew}). Therefore, any point $(\pi, \phi)$ belongs to $\mathcal{A}_i$ if
\begin{equation}
 \left\{\begin{array}{l}
         i = \arg\max_j r(\pi, \phi, j) \\
         r(\pi, \phi, i) > 0
        \end{array}
\right.
\end{equation}
The first conditions states that the best transmission level is $i$, while the second ensures that transmitting is better than deferring.
Now, we prove by contradiction that any region $\mathcal{A}_i$, with $1\leq i\leq N-1$ can be adjacent only to the regions $\mathcal{A}_{i-1}$ and $\mathcal{A}_{i+1}$.
Since $r(\pi, \phi, i)$ is continuous $\forall i$, if there exists a boundary $\mathcal{B}\subset\mathcal{U}$ between $\mathcal{A}_i$ and $\mathcal{A}_{i+j}$, with $j>1$, it would follow that, for any $(\pi,\phi)\in\mathcal{B}$, $r(\pi,\phi,i) = r(\pi,\phi,i+j) > r(\pi,\phi,h)$, $\forall h\neq i,i+j$. In particular, if we chose $h = i+1$, we would have $r(\pi,\phi,i) = r(\pi,\phi,i+j) > r(\pi,\phi,i+1)$.

However, this is not possible. In fact, the curve $r(\pi,\phi,x)$, with $x\in\mathbb{R}^+$, either has a unique global maximum at
\begin{equation}
 x^* = 1 + \log_2\left(\frac{(c\phi+d) - (a\pi+b)}{\ln(c\phi+d) - \ln(a\pi+b)-\ln(k)}\right)
\end{equation}
or it is strictly monotonic over all its domain\footnote{it can be easily proved that the values $(\pi,\phi)$ where $r(\pi, \phi, x)$ is monotonically decreasing in $x$ all belong to $\mathcal{A}_0$.}.

It follows that if $r(\pi,\phi,i) = r(\pi,\phi,i+j)$, then $x^*$ exists in the interval $(i, i+j)$; moreover, $r(\pi,\phi,x)$ is greater than both $r(\pi,\phi,i)$ and $r(\pi,\phi,i+j)$ for any $x\in(i, i+j)$, and therefore also for $x = i+1$, which contradicts the hypothesis. This proves the lemma.
\end{proof}

\section{Proof of Lemma \ref{lem:boundaries}}
\label{app:proof4}
\begin{proof}
 We call $\mathcal{G}_i$ the subset of $\mathcal{U}$ where the reward obtained with power level $i$ is equal to that attained with power level $i+1$. Using (\ref{reweq}), this is equivalent to set $r(\pi,\phi,i)=r(\pi,\phi,i+1)$, resulting in the curve $\mathcal{G}_i(\pi,\phi)$ defined by the equation:
 \begin{equation}
 \exp\left(-\frac{a\pi+b}{2^{i-1}}\right) - \exp\left(-\frac{a\pi+b}{2^i}\right) = \frac{1}{k}\left(\exp\left(-\frac{c\phi+d}{2^{i-1}}\right) - \exp\left(-\frac{c\phi+d}{2^i}\right)\right)
 \label{geng}
\end{equation}
As observed in Lemma \ref{lem:many_regions}, at any point $(\pi,\phi)\in\mathcal{G}_i$, trasmitting with power level $i$ or $i+1$ gives the same reward, which is the highest achievable if a transmission is performed. However, in order for $(\pi,\phi)$ to be on the boundary between $\mathcal{A}_i$ and $\mathcal{A}_{i+1}$, it must also be that $r(\pi,\phi,i)=r(\pi,\phi,i+1)>0$, meaning that $(\pi,\phi)\notin\mathcal{A}_0$.
Therefore, the boundary between $\mathcal{A}_i$ and $\mathcal{A}_{i+1}$ corresponds to $\mathcal{G}_i^+ = \mathcal{G}_i\setminus\mathcal{A}_0$.

Equation (\ref{geng}) can be explicited as follows. If $k<1$, then we can rewrite it as the union of two continuous functions $g_i^-(\pi)$ and $g_i^+(\pi)$:
\begin{equation}
 g_i^{\pm}(\pi) = -\frac{d}{c} - \frac{2^i}{c}\ln\left(\frac{1}{2}\mp\frac{1}{2}\sqrt{1-4k\exp\left(-\frac{a\pi+b}{2^i}\right)\left(1-\exp\left(-\frac{a\pi+b}{2^i}\right)\right)}\right)
 \label{defg}
\end{equation}
which are both defined over all the domain $\mathbb{R}^+$. It can be shown mathematically that $g_i^-(\pi) < g_0(\pi) < g_i^+(\pi)$, $\forall i\in\{1,2,\ldots,N\}$ and $\forall\pi\in\mathbb{R}^+$. This means that $g_i^-(\pi)\subset\mathcal{A}_0$, whereas $g_i^+(\pi)\cap\mathcal{A}_0=\emptyset$.
Therefore, the only boundary $\mathcal{G}_i^+$ between $\mathcal{A}_i$ and $\mathcal{A}_{i+1}$ is the continuous function $g_i^+(\pi)$.

Conversely, if $k>1$, then equation (\ref{geng}) can be reformulated as:
\begin{equation}
 h_i^{\pm}(\phi) = -\frac{b}{a} - \frac{2^i}{a}\ln\left(\frac{1}{2}\pm\frac{1}{2}\sqrt{1-\frac{4}{k}\exp\left(-\frac{c\phi+d}{2^i}\right)\left(1-\exp\left(-\frac{c\phi+d}{2^i}\right)\right)}\right)
 \label{defh}
\end{equation}
which is again the union of two continuous functions defined over the entire domain $\mathbb{R}^+$. To complete the proof, we recall that, since $g_0(\pi)$ is monotonically increasing, it can be inverted into $h_0(\phi)$, and consequently $\mathcal{A}_0$ can be equivalently defined by the condition $\pi>h_0(\phi)$.
Setting the inequality shows that $h_i^+(\phi) < h_0(\phi) < h_i^-(\phi)$, $\forall i\in\{1,2,\ldots,N\}$ and $\forall\phi\in\mathbb{R}^+$. This implies that $h_i^-(\phi)\subset\mathcal{A}_0$ while $h_i^+(\phi)\cap\mathcal{A}_0=\emptyset$, which in turns means that the only boundary $\mathcal{G}_i^+$ between $\mathcal{A}_i$ and $\mathcal{A}_{i+1}$ is the continuous function $h_i^+(\phi)$.
\end{proof}

\section{Algorithm to compute $k$}
\label{app:algo}
The optimal strategy, once the topology is known, is completely defined by the value $k$, which appears in the definition of all the boundary functions $g_0(\pi)$, $g_i^+(\pi)$ and $h_i^+(\phi)$.
When multiple power levels are available, the derivation of $k$ can be done numerically. The first option is to use the iterative algorithm in (\ref{updapi}) and (\ref{updaV}). This algorithm is proved to converge to the optimal solution, but the convergence time rapidly grows with the number of states, which in turn depends on the quantization step adopted for $\pi$ and $\phi$ which are, respectively, the power received from $U$ at $D$ and $B$.
In fact, in this algorithm the future cost at any state $(\pi,\phi)$ is updated at each iteration, starting from 0. If $\gamma$ is close to 1, it may take several iterations before the value of the future cost achieves its asymptotic value.

A different way is based on the fact that the cost of the punishment is the same for any state $(\pi,\phi)$, and is equal to $1/k$, as per (\ref{funrew}). If we set an initial value for this cost, it is possible to refine this value at each step, until it is consistent with the overall expected reward $C$.

In other words, at each step $t$, we use the current value $C_t$ and the corresponding $k_t = (\gamma C_t(1-\gamma^W))^{-1}$ to determine the expected reward for any state $(\pi,\phi)$ and for any power level $i$ as per (\ref{reweq}):
\begin{equation}
 r(\pi, \phi, i) = \gamma C_t + p_i(\pi) - \frac{1}{k(t)}q_i(\phi)
\end{equation}
where, adapting (\ref{defpq}) to this specific scenario,
\begin{equation}
 p_i(\pi) = \exp\left(-\frac{\theta}{\gamma_{sb}^{(i)}}\frac{\pi}{N_0}\right)\exp\left(-\frac{\theta}{\gamma_{sd}^{(i)}}\right) \quad\quad  q_i(\phi) = \min\left(\exp\left(\frac{1-\phi/(\theta N_0)}{\gamma_{sb}^{(i)}}\right), 1\right)
\end{equation}
Notice that both $p_i(\pi)$ and $q_i(\phi)$ are to be computed only at the beginning of the algorithm.

Now, for any state $(\pi,\phi)$, we compute the maximum reward:
\begin{equation}
 r^*(\pi, \phi) = \max_{i\in\{1,2,\ldots, N\}}r(\pi,\phi,i)
\end{equation}
where $N$ is the total number of available power levels. If $r^*(\pi,\phi)<0$, then the best option is not to transmit at all, and we set correspondingly $r^*(\pi,\phi)=0$.

Having obtained the optimal reward for any state, we derive the new overall expected reward by numerically integrating:
\begin{equation}
C_{t+1} = \sum_{(\pi,\phi)\in\mathcal{S}}r^*(\pi,\phi)\phi_{\Pi,\Phi}(\pi,\phi)
\end{equation}
The obtained value of $C_{t+1}$ is finally used to derive the new value $k_{t+1}$:
\begin{equation}
 k_{t+1} = \frac{1}{\gamma C_{t+1}(1-\gamma^W)}
\end{equation}
A new iteration of the algorithm can therefore be performed, until $k_t$ converges with the desired precision to the effective value $k^*$.
\begin{algorithm}
\caption{Computation of $k$}
 \begin{algorithmic}
 \State $k \gets 1$
 \State $C \gets \frac{1}{\gamma k(1-\gamma^W)}$
 \State $\Delta \gets \infty$
 \While {$\Delta > t$}
 \For {$\pi\in\Pi, \phi\in\Phi$}
    \For {$P_i\in\mathcal{P}$}
      \State $r(\pi,\phi,i) \gets \gamma C + p_i(\pi) - q_i(\phi)/k$
    \EndFor
    \State $r^*(\pi,\phi) \gets \max_{i\in\{1,2,\ldots, N\}}r(\pi,\phi,i)$
    \State $r^*(\pi,\phi) \gets \max\left(r^*(\pi,\phi),0\right)$
  \EndFor
  \State $C \gets \sum_{(\pi,\phi)\in\Pi\times\Phi}r^*(\pi,\phi)\phi_{\Pi,\Phi}(\pi,\phi)$
  \State $k_t \gets \frac{1}{\gamma C(1-\gamma^W)}$
  \State $\Delta \gets |k-k_t|$
  \State $k\gets k_t$
  \EndWhile
  \Return $k$
 \end{algorithmic}
\end{algorithm}

\bibliographystyle{IEEEtran}
\bibliography{IEEEabrv,biblio}

\end{document}